\documentclass[journal,twoside,final]{IEEEtran}

\usepackage{comment}
\usepackage{url}

\usepackage{setspace}  
\newcounter{MYtempeqncnt}
\pdfoutput=1

\usepackage{color}
\usepackage{ifpdf}
\usepackage{epsfig}
\usepackage{graphicx}
\usepackage{amsmath}
\usepackage{amsthm}
\usepackage{amsfonts}
\usepackage{amssymb}
\usepackage{eucal}
\usepackage{subfigure}
\usepackage{array}
\usepackage{float}
\usepackage{url}
\usepackage{multirow}
\usepackage[numbers,sort&compress]{natbib}
\usepackage{enumerate}
\newtheorem{theorem}{Theorem}
\newtheorem{lemma}{Lemma}
\newtheorem{remark}{Remark}
\newtheorem{Proposition}{Proposition}

\usepackage{tikz}
\usetikzlibrary{arrows}
\usetikzlibrary{arrows,positioning}
\usetikzlibrary{decorations.markings}
\tikzset{
    >=stealth',
    punkt/.style={
           rectangle,
           rounded corners,
           draw=black, very thick,
           text width=6.5em,
           minimum height=2em,
           text centered},
    pil/.style={
           ->,
           thick,
           shorten <=2pt,
           shorten >=2pt,}
    pir/.style={
           <-,
           thick,
           shorten <=2pt,
           shorten >=2pt,}
}
\definecolor{KeynoteRed}{rgb}{.678,.051, .051}
\definecolor{KeynoteBlue}{rgb}{0.008, 0.443, 0.60}
\definecolor{KeynoteLightblue}{rgb}{.635, .914, .973}
\definecolor{KeynoteYellow}{rgb}{0.859, 0.584, 0.212}
\definecolor{KeynoteYellow}{rgb}{0.859, 0.584, 0.212}
\definecolor{KeynoteSlate}{rgb}{0.239, 0.271, 0.322}
\definecolor{KeynoteGray}{rgb}{0.498, 0.529, 0.529}
\definecolor{KeynoteTextGray}{rgb}{0.325, 0.325, 0.325}
\definecolor{KeynoteLightGray}{rgb}{0.706, 0.706, 0.706}
\definecolor{KeynoteBlueGray}{rgb}{0.471, 0.533, 0.620}
\definecolor{ECEpurple}{rgb}{.169, .18, .455}
\definecolor{ECEcyan}{rgb}{.41, .62, .72}
\definecolor{ECEgray}{rgb}{.788, .827, .859}
\definecolor{ECEblueGray}{rgb}{61.2, 70.6, 70.6}
\definecolor{ECEblueGray}{rgb}{61.2, 70.6, 70.6}
\definecolor{RiceBlue}{rgb}{0, .14, .41}

\usepackage{flushend}
\begin{document}

\title{Distributed Full-duplex via Wireless Side-Channels: Bounds and Protocols}
\author{
      Jingwen Bai and Ashutosh Sabharwal, \emph{Fellow}, IEEE
\thanks{%
Manuscript received December 20, 2012; revised March 29, 2013; accepted May 25, 2013. The associate
editor coordinating the review of this paper and approving it for
publication was Prof. S. Valaee.}
\thanks{%
The authors are with Department of Electrical and Computer Engineering
      Rice University, Houston, TX 77005, USA (e-mail: \{jingwen.bai, ashu\}@rice.edu).}
\thanks{%
This work was partially supported by NSF CNS-1012921 and NSF CNS-1161596.}
\thanks{%
Digital Object Identifier 10.1109/TWC.2013.122015.}}


\maketitle

\markboth{IEEE Transactions on Wireless Communications, Vol. 12, No.
8, Month 2013}{Bai \MakeLowercase{and} Sabharwal:
 Distributed Full-duplex via Wireless Side-Channels: Bounds and Protocols}

\pubid{1536-1276/12\$31.00~\copyright~2013 IEEE}

\pubidadjcol

\begin{abstract}
 In this paper, we study a three-node full-duplex network, where a base station is engaged in simultaneous up- and downlink communication in the same frequency band with two half-duplex mobile nodes. To reduce the impact of inter-node interference between the two mobile nodes on the system capacity, we study how an orthogonal side-channel between the two mobile nodes can be leveraged to achieve full-duplex-like multiplexing gains. We propose and characterize the achievable rates of four distributed full-duplex schemes, labeled bin-and-cancel, compress-and-cancel, estimate-and-cancel and decode-and-cancel.
 Of the four, bin-and-cancel is shown to achieve within 1~bit/s/Hz of the capacity region for all values of channel parameters. In contrast, the other three schemes  achieve the near-optimal performance only in certain regimes of channel values. Asymptotic multiplexing gains of all proposed schemes are derived to show that the side-channel is extremely effective in  regimes where inter-node interference has the highest impact.
\end{abstract}

\section{Introduction}
\PARstart{C}{urrent} deployed wireless communication systems adopt either time-division or frequency-division for transmission and reception. Recently, many researchers have implemented full-duplex wireless communication systems which support simultaneous transmission and reception in the same frequency band. In these system level implementations, different self-interference cancellation mechanisms have been proposed, including passive self-interference suppression~(e.g., \cite{evan,sahai2011pushing,duarte2012design}) and active self-interference cancellation (e.g., \cite{Choi,duarte2011experiment,radunovic2010rethinking,duarte2012design}). While the first set of experiments were all short-range, recent systems which employ antenna isolation at infrastructure nodes~\cite{Everett12Thesis} have reported full-duplex communication ranges in the order of 100s of meters, which are suitable for small cells~\cite{Smallcell}.
\begin{figure}[htbp]
\begin{center}
\input{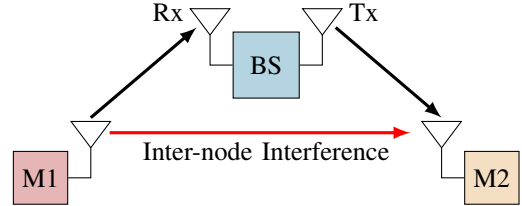}
\caption{\emph{Three-node full-duplex network: inter-node interference becomes an important factor when the infrastructure node communicates with uplink and downlink mobile nodes simultaneously.}}
\label{fig.1}
\end{center}
\end{figure}
The feasibility of full-duplex in short to medium ranges opens new design opportunities for wireless communication systems. We envision that the first adoption of full-duplex will be in small cell infrastructure nodes, since infrastructure nodes have more flexibility in their physical size and shape compared to mobile nodes.  Thus, we expect that small form-factor mobile nodes will continue to be half-duplex in the near future. Within this context, we focus on studying the performance of a network of  half-duplex mobiles being served by a full-duplex base station~(BS).

Consider the three-node network shown in Figure~\ref{fig.1}. A full-duplex infrastructure node communicates with two half-duplex mobiles simultaneously to support one uplink and one downlink flow. The network has to deal with two forms of interferences: self-interference at the full-duplex infrastructure node, and inter-node interference (INI) from uplink mobile~M1 to downlink mobile~M2. The two forms of interferences differ in one important aspect: the self-interfering signal is known at the interfered receiver, since the transmitter and receiver are co-located at BS, but the INI is not known at the receiver Node~M2. In a small cell scenario, we assume that the self-interference at the infrastructure has been suppressed to the noise floor (see discussion in first para, in~\cite{Everett12Thesis}) and hence focus on INI. To the best of the authors' knowledge, this paper is the first to study the impact of INI in the three-node full-duplex network and the mitigation of INI using advanced interference management techniques.\pubidadjcol

In this paper, we study how a wireless side-channel between Nodes~M1 and~M2 can be leveraged to manage INI. Conceptually, one can model the co-location of transmitter and receiver on a full-duplex node as an infinite capacity side-channel. Thus, our use of a wireless side-channel between~M1 and~M2 mimics the inherent full-duplex side-channel but has finite bandwidth and signal-to-noise ratio~(SNR) like any practical wireless channel. The wireless side-channel model is inspired by the fact that most smartphones support simultaneous use of multiple standards, and can thus access multiple orthogonal spectral bands. For example, the main network could be on a cellular band while the M1$\rightarrow$M2 wireless side-channel could be on an unlicensed ISM band, thus creating a novel use of ISM-bands in licensed cellular networks.
Based on the above discussion, we label the protocols for communication in side-channel assisted three-node network as \emph{distributed full-duplex}.

Our contributions in this paper are two-fold. First, we propose four distributed full-duplex inter-node interference cancelation schemes. The four schemes are labeled bin-and-cancel~(BC), compress-and-cancel~(CC), decode-and-cancel~(DC)~\cite{Jingwen} and estimate-and-cancel~(EC). All schemes rely on the side-channel to send information about the INI from Node~M1 to Node~M2. Since the side-channel is an orthogonal channel, Node~M1 uses the side-channel signal as side information while decoding its signal of interest from BS. As the names suggest, all four schemes encode the M1 signal on the side-channel in different ways. The most sophisticated of the four, bin-and-cancel\footnote{A similar approach is also adopted in \cite{Zchannel} for Z-interference channel with a relay link.}, uses Han-Kobayashi style common-private message splitting and can achieve within 1~bit/s/Hz of the capacity region for \emph{all} values of channel parameters. The other three schemes are simpler compared to bin-and-cancel but achieve the near-optimal (finite approximation) performance only in certain regimes of channel values; we derive exact regions of approximate optimality for decode- and estimate-and-cancel.

Second, we derive the asymptotic multiplexing gains of bin-, compress-, decode-, estimate-and-cancel schemes. We show analytically that the side information can be highly beneficial in increasing the multiplexing gain of the system exactly in those regimes where INI has the highest impact. We provide an exact characterization of how the extra bandwidth of the side-channel can be leveraged to achieve the multiplexing gains. As expected, the multiplexing gain scales with the bandwidth of the side-channel and can reach the maximum value of 2, signifying perfect full-duplex. Finally, we show numerically that significant multiplexing gains are available for finite SNRs of practical interest, and thus our analysis makes a case for using wireless spectrum in a more flexible manner (e.g\ unlicensed bands in licensed spectrum).

The rest of paper is organized as follows. Section 2 presents the system model.
In Section 3, we propose four distributed full-duplex interference cancellation schemes by leveraging the side-channel and obtain information theoretic results on their achievable rate regions for Gaussian channels. Results on the capacity region of side-channel assisted three-node network are also presented in this section. Section 4 characterizes multiplexing gains of all proposed schemes in high SNR limit and finite SNR regime.
Section 5 concludes the paper.

\section{System Model \label{sec:model}}
In this section, we will describe the system model to be used for the rest of the paper. Since we assume that the self-interference at the full-duplex base station is below the noise floor, the resulting network can be modeled as a Z-like channel with an orthogonal side-channel as shown in Figure~\ref{fig.10}.
Each node is assumed to have a bounded power budget, so we denote the power constraint at BS and Node~M1 as $P_1$ and $P_2$, respectively. Further let $W_m$ denote the bandwidth of the main-channel, and let $W_s$ denote the bandwidth of the side-channel.
Parameter $W=\frac{W_s}{W_m}$ represents the bandwidth ratio of the side-channel to that of the main-channel, where $W\in[0,W_{\rm max}]$.
\begin{figure}[htbp]
\begin{center}
\input{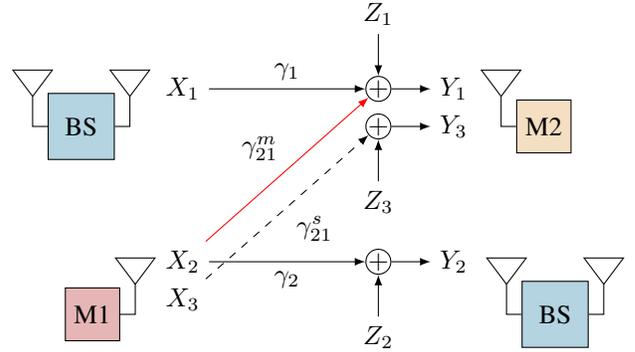}
\caption{\emph{System model: Z-like channel with side-channel.}}
\label{fig.10}
\end{center}
\end{figure}

The following equations capture the input-output relationship between transmitted and received signals:
\begin{equation}\label{systemequ}
\begin{aligned}
Y_1=&\gamma_1X_1+\gamma^m_{21}X_2+Z_1 \\
Y_2=&\gamma_2X_2+Z_2   \\
Y_3=&\gamma^s_{21}X_3+Z_3, 
\end{aligned}
\end{equation}
where for $i=1,2,3$, $X_i$ is the codeword with input power constraint $\mathbb{E}(|X_1|{}^2)\leq P_1$, $\mathbb{E}(|X_2+X_3|^2) \leq P_2$. Let $Z_i$ be independently and identically distributed~(i.i.d.) white Gaussian noise with zero mean and variance of $\sigma^2 $.
The coefficients $\gamma_1, \gamma_2$ are the channel gains of the direct link from the transmitters to their intended receivers,  $\gamma^m_{21}$ is the channel gain of the interference link from M1 to M2, and $\gamma^s_{21}$ is the channel gain of the side-channel from M1 to M2.\footnote{All channels are assume to be real for simplicity and all results can be extended to complex channels.}
To simplify the notation, let ${\rm SNR_1}=\frac{|\gamma_1|^2 P_1}{\sigma^2}, {\rm SNR_2}=\frac{|\gamma_2|^2P_2}{\sigma^2}, {\rm INR}=\frac{|\gamma^m_{21}|^2P_2}{\sigma^2}$ and ${\rm SNR_{\rm side}}=\frac{|\gamma^s_{21}|^2P_2}{\sigma^2}$.

\section{Four Schemes for Side-channel Assisted Three-node Network}
In this section, we propose four distributed full-duplex interference cancellation schemes by using a wireless side-channel between the two mobile nodes.
The four schemes are labeled bin-and-cancel~(BC), compress-and-cancel~(CC), decode-and-cancel~(DC) and estimate-and-cancel~(EC).
Each scheme leverages the side-channel in a different manner to send information about INI from Node~M1 to Node~M2. Since the side-channel operates on an
orthogonal band compared to the main-channel, Node~M2 can use the side information obtained from the side-channel while decoding its signal of interest from BS.

\subsection{Bin-and-cancel} \label{BCscheme}
In bin-and-cancel, we use Han-Kobayashi \cite{han1981new} style common-private message splitting.
The main idea behind BC scheme is to divide the interfering message of Node~M1 into two parts: private message which can only be decoded
at the intended receiver Node~BS, and common message which can be decoded at both receivers. Next, we partition the common message of M1 into equal size bins, encode the bin indexes into codewords and send through the side-channel. The number of bin indexes is determined by the rate of the side-channel.
At the interfered receiver M2, by decoding the common message of the interfering signal from M1, part of the
interference can be subtracted out while treating the remaining private message of M1 as noise. The bin index can first be decoded from the side-channel. Then with the help of the bin index, the uncertainty of decoding the common message of M1
can be resolved, allowing sending more common message of M1 and mitigating INI.

In Gaussian channels, we first split the total transmit power of Node~M1 between side-channel and main-channel, then split the power allocated to
the main-channel between the common and private message. Assuming inputs are i.i.d. Gaussian distributed satisfying the per-node power constraint, we have the following achievability result.
\begin{figure}[t!] 
  \centering
    \includegraphics[width=0.5\textwidth,trim = 70mm
      85mm 40mm 60mm, clip]{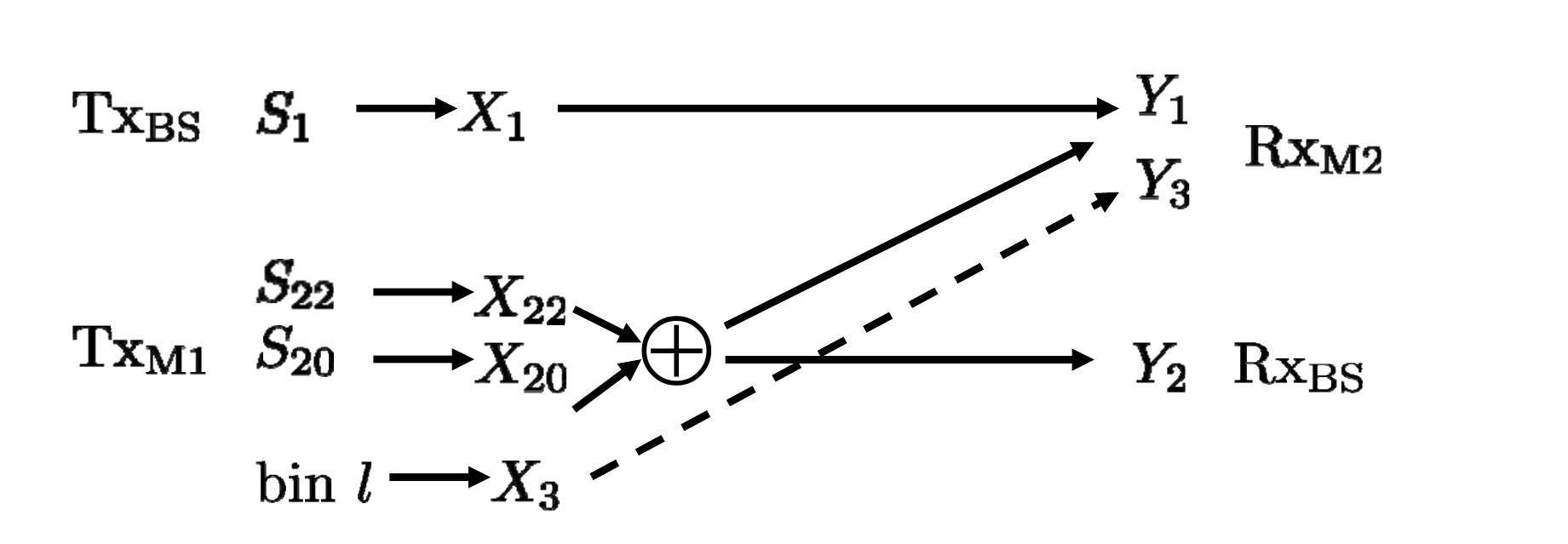}
  \caption{\emph{Depiction of bin-and-cancel.}}
\label{BC-diagram.}
\end{figure}
\begin{Proposition}\label{BCrate}For Gaussian side-channel assisted three-node full-duplex network, in the weak interference regime defined by $\rm INR\leq SNR_2$,
\footnote{When $W=0$, we define $W\mathrm{log}\left(1+\frac{x}{W}\right)\triangleq0$.} the following rate region is achievable for bin-and-cancel,
\begin{gather}
\begin{aligned}
R_{1}\leq& C\left(\frac{\mathrm{SNR_1}}{1+\beta\bar{\lambda}\rm INR}\right)\\
R_2\leq& \min\!\left\lbrace C\left(\bar{\lambda}\mathrm{SNR_2}\right),C\left(\beta\bar{\lambda}\mathrm{SNR_2}\right)\right.\\
&\left.+C\left(\frac{\bar{\beta}\bar{\lambda}\mathrm{INR}}{1+\beta\bar{\lambda}\mathrm{INR}}\right)\!+\!WC\left(\frac{\mathrm{\lambda SNR_{side}}}{W}\right)\!\right\rbrace\\
R_1+R_2\leq&C\left(\beta\bar{\lambda}\mathrm{SNR_2}\right)+C\left(\frac{\mathrm{SNR_1}+\bar{\beta}\bar{\lambda}\mathrm{INR}}{1+\beta\bar{\lambda}\mathrm{INR}}\right)\\
&+WC\left(\frac{\lambda \mathrm{SNR_{side}}}{W}\right),
\label{bc-sum}
\end{aligned}
\end{gather}
where $\beta, \lambda\in[0,1], \beta+\bar{\beta}=1, \lambda+\bar{\lambda}=1$ and $C(X)=W_m{\rm log}\left(1+X\right)$.

In the strong interference regime defined by $\rm SNR_2\leq INR\leq SNR_2\left(1+SNR_1\right)$, the achievable rate region is
\begin{gather}
\begin{aligned}
 R_{1}\leq& C\left(\rm SNR_1\right)\\
R_2\leq&  C\left(\bar{\lambda}\rm SNR_2\right)\\
R_1+R_2\leq&C\left(\mathrm{SNR_1}+\bar{\lambda}\mathrm{INR}\right)+WC\left(\frac{\mathrm{\lambda SNR_{side}}}{W}\right). \label{BCstrong}
\end{aligned}
\end{gather}

In the very strong interference regime defined by $\rm INR\geq SNR_2\left(1+SNR_1\right)$, the capacity region is
\begin{gather}
\begin{aligned}
  R_{1}\leq& C\left(\rm SNR_1\right)\triangleq C_1\\
   R_{2}\leq& C\left(\rm SNR_2\right)\triangleq C_2 \label{PTP}
\end{aligned}
\end{gather}
\end{Proposition}
\begin{proof}
The encoding procedure is depicted in Figure~\ref{BC-diagram.}. Using the Han-Kobayashi common-private message splitting strategy, the common message
can be decoded at both receivers, while private message can only be decoded at the intended receiver. Message $S_1$ is of size $2^{nR_1}$, and $X_1$
is intended to be decoded at $Y_1$ only. Message $S_2$ is divided into common part $S_{20}$ of size $2^{nR_{20}}$, and private part $S_{22}$ of
size $2^{nR_{22}}$. Superpose the codewords of both $S_{22}$ and $S_{20}$ such that $X_2=X_{20}+X_{22}$.
Then partition the set
$[1:2^{nR_{20}}]$ into $2^{nR_3}$ equal size bins, $\mathcal{B}(l)=[(l-1)2^{n(R_{20}-R_3)}+1:l2^{n(R_{20}-R_3)}], l\in[1:2^{nR_3}]$, and $X_3(l)$
is sent through the side-channel over $n$ blocks. In the Gaussian channels, we assume all inputs are i.i.d. Gaussian distributed satisfying $X_1\sim\mathcal{N}(0,P_1), X_{20}\sim\mathcal{N}(0,\bar{\beta}\bar{\lambda}P_2)$, $X_{22}\sim\mathcal{N}(0,\beta\bar{\lambda}P_2)$,
and $X_{3}\sim\mathcal{N}(0,\lambda P_2)$, respectively, where $\beta, \lambda\in[0,1], \bar{\beta}+\beta=1, \bar{\lambda}+\lambda=1$.
The parameter $\lambda$ denotes the fraction of power allocated to the side-channel and $\beta$ denotes the fraction of power split for the private message of M1.

Decoding occurs in two steps. Upon receiving $Y_2, (S_{20},S_{22})$ are first decoded. The achievable rate region of $(R_{20},R_{22})$ is the capacity region of a Gaussian multiple-access channel denoted as $\mathcal{C}_1$, where 
\begin{gather}
\begin{aligned}
R_{20}&\leq C\left(\bar{\beta}\bar{\lambda}\mathrm{SNR}_2\right)\\
R_{22}&\leq C\left(\beta\bar{\lambda}\mathrm{SNR}_2\right)\\
R_{20}+R_{22}&\leq C\left(\bar{\lambda}\mathrm{SNR}_2\right). \label{BCequ1}
\end{aligned}
\end{gather}

Then upon receiving $Y_1, Y_3$, with the help of the side-channel $(S_1,S_{20})$ are decoded treating $S_{22}$ as noise. This is a multiple-access channel with side-channel.
The capacity region of such channel denoted as $\mathcal{C}_2$ is given in Appendix~\ref{Append1}. Thus we have
\begin{gather}
\begin{aligned}
R_{1}&\leq C\left(\frac{\mathrm{SNR_1}}{1+\beta\bar{\lambda}\rm INR}\right)\\
R_{20}&\leq C\left(\frac{\bar{\beta}\bar{\lambda}\mathrm{INR}}{1+\beta\bar{\lambda}\rm INR}\right)\!+\!WC\left(\frac{\mathrm{\lambda SNR_{side}}}{W}\!\right) \\
R_{1}+R_{20}&\leq C\left(\frac{\mathrm{SNR_1}+\bar{\beta}\bar{\lambda}\mathrm{INR}}{1+\beta\bar{\lambda}\rm INR}\right)\!+\!WC\left(\frac{\mathrm{\lambda SNR_{side}}}{W}\!\right).\label{BCequ2}
\end{aligned}
\end{gather}
The achievable rate region of Gaussian side-channel assisted three-node network is the set of all $(R_1,R_2)$ such that $R_1, R_2=R_{20}+R_{22}$ satisfying that $(R_{20}, R_{22})\in\mathcal{C}_1$ and $(R_1, R_{20})\in\mathcal{C}_2$. Using Fourier-Motzkin elimination, we can derive the results in Proposition~\ref{BCrate}.

In the strong interference regime where $\rm SNR_2\leq INR\leq SNR_2\left(1+SNR_1\right)$, it is not difficult to find out the achievable rate region is a decreasing function of $\beta$~(see e.g.~\cite{Zchannel}). Therefore, we can obtain the achievable region by substituting $\beta=0$.

When the interference is very strong where $\rm INR\geq SNR_2\left(1+SNR_1\right)$, Carleial~(\cite{Car}) shows that the communication is not impaired by the very strong interference at all. Because the interfered receiver can first decode the interfering signal, then subtract it out from the received signal and finally decode its own message. Thus the capacity region is contained by the point-to-point capacity of uplink and downlink.
\end{proof}
\begin{figure*}[!t]
\normalsize
\setcounter{MYtempeqncnt}{\value{equation}}
\setcounter{equation}{6}
\begin{equation}
\begin{aligned}
\label{ccrate}
R_1&\leq C\left(\frac{\mathrm{SNR_1}\left(1+\mathrm{SNR_1}+(1+\mathrm{SNR_1}+\bar{\lambda}\mathrm{INR})\left[\left(1+\frac{\lambda \rm SNR_{side}}{W}\right)^W-1\right]\right)}{(1+\mathrm{SNR_1}+\bar{\lambda}\mathrm{INR})\left[\left(1+\frac{\lambda \rm SNR_{side}}{W}\right)^W-1\right]+(1+\mathrm{SNR_1})(1+\bar{\lambda}\rm INR))}\right) \\
R_2&\leq C\left(\bar{\lambda}\rm SNR_2\right).
\end{aligned}
\end{equation}
\setcounter{equation}{\value{MYtempeqncnt}}
\hrulefill
\vspace*{4pt}
\end{figure*}
\begin{remark}
In the weak interference regime, BC scheme contributes to improving the downlink rate $R_1$ which is limited by the interference
from M1. While in the strong
interference regime where $\mathrm {INR\geq SNR_2}$, M1 sends all common message, and there is no interference at the receiver of M2. In this case, BC scheme
enables M1 to deliver more common message which is restricted by the interference link, thus increasing the uplink rate $R_2$.
\end{remark}
\begin{remark}
In~\cite{Zchannel}, the authors adopted a similar approach for Z-interference channel with a digital relay link which is error-free between the receivers. Our system
model and the system model in \cite{Zchannel} will reduce to the Gaussian Z-interference channel when there is no side-channel between the interfered transceiver or there is no relaying link
between the two receivers.
However, in our system model, the side-channel is a wireless channel which exists between the two mobile nodes. This is motivated by the prevalence
of multi-radio interfaces on current mobile devices.
Also, we consider the power allocation between main-channel and side-channel for the proposed scheme to meet the per-node power constraint.
In contrast, in \cite{Zchannel} there is no such power constraint since the relay schemes are employed between the receivers over a noiseless relay link.
\end{remark}
BC scheme adopts Han-Kobayashi strategy which allows arbitrary splits of total transmit power between common message and private message
in the weak interference regime
in addition to the power split between the side-channel and main-channel.
Therefore, the optimization among such myriads of possibilities is hard in general.
We can simplify BC scheme by setting power splitting of private message
of M1 at the level of Gaussian noise at the receiver M2 in the weak interference regime.
This particular simple common-private power splitting is used in \cite{D} and
is shown to achieve to within one bit of the capacity region of the general two-user interference channel.
We will use the simplified BC scheme with fixed power splitting for the capacity analysis in Section~\ref{sec:one-bit}.
\subsection{Three Simpler Schemes \label{sec:simpler}}
We propose three simpler schemes compared to BC scheme for encoding $X_3$ and sequentially post-processing $\left(Y_1,Y_3\right)$ at receiver M2.
The first is compress-and-cancel, which simplifies the transmitter design of Node~M1. Then we give another two schemes: decode-and-cancel and estimate-and-cancel,
both of which not only simplify the transmitter design of M1,
but also the receiver design of M2.
\subsubsection{Compress-and-cancel}
In compress-and-cancel scheme, the transmitter design of M1 is simplified compared to BC scheme. Node~M1 sends a compressed version of
the interfering message over the side-channel using source coding.
Noticing that the correlated information with the interfering message can be observed at receiver M2,
Wyner-Ziv strategy~\cite{wyner} can be adopted for compression. We first compress $X_2$ into $\widehat{X_2}$, then
encode $\widehat{X_2}$ as $X_3$ and transmit it over the side-channel.
At the receiver M2, with the knowledge of the distribution of $Y_1$,  we can recover the interference
under certain distortion and subtract it from the main-channel. The capacity of the side-channel should allow reliable transmission
of the compressed signal $\widehat{X_2}$.

In the Gaussian channels, all inputs are assumed to be i.i.d. Gaussian distributed, satisfying $X_1\sim\mathcal{N}(0,P_1), X_2\sim\mathcal{N}(0,\bar{\lambda}P_2)$
and $X_{3}\sim\mathcal{N}(0,\lambda P_2)$, respectively, where $\lambda\in[0,1],\bar{\lambda}+\lambda=1$.
For the ease of analysis, we also assume $\widehat{X_2}$ to be a Gaussian quantized version of $X_2$:
\begin{gather}
\begin{aligned}
\widehat{X_2}=X_2+Z_q,
\end{aligned}
\end{gather}
where $Z_q$ is the quantization noise with Gaussian distribution $\mathcal{N}(0,\sigma_q^2)$.
We can obtain the following inner bound, with calculation of rate in Appendix~\ref{Append2}.
\addtocounter{equation}{2}
\begin{Proposition}For Gaussian side-channel assisted three-node full-duplex network,the achievable rate region for compress-and-cancel is given in (\ref{ccrate}) at the top of the page.
\end{Proposition}
\subsubsection{Decode-and-cancel}
Decode-and-cancel can lead to a simpler transceiver design for M1 compared to BC scheme because no common-private message splitting is adopted at
the transmitter M1, and only single-user decoders are involved at the receiver M2.
In DC scheme \cite{Jingwen}, the interfering message $S_2$ of Node~M1 is encoded into two codewords $X_2$ and $X_3$ using two independent Gaussian codebooks. Then $X_2$ is sent through the main-channel and $X_3$ is sent through the side-channel. The message $S_2$ is required to be decoded at both the intended receiver BS and the side-channel receiver at M2. After decoding $S_2$ from the side-channel,
the interference can be cancelled out at the main-channel receiver at M2.

In the Gaussian channels, all inputs are assumed to be i.i.d. Gaussian distributed, satisfying $X_1\sim\mathcal{N}(0,P_1), X_2\sim\mathcal{N}(0,\bar{\lambda}P_2)$
and $X_{3}\sim\mathcal{N}(0,\lambda P_2)$, respectively, where $\lambda\in[0,1],\bar{\lambda}+\lambda=1$.
The parameter $\lambda$ denotes the power allocated to the side-channel. We obtain the following inner bound achieved by DC.
\begin{Proposition}(from \cite{Jingwen}) For Gaussian side-channel assisted three-node full-duplex network, the following rate region is achievable for decode-and-cancel
\setcounter{equation}{7}
\begin{align}
 R_{1}& \leq C_1\nonumber \\
R_2&\leq\min \left\{\!W C\left(\!\frac{\lambda \rm SNR_{side}}{W}\right),C\left(\bar{\lambda}\rm SNR_2\right)\!\right\}
\end{align}
\end{Proposition}
\subsubsection{Estimate-and-cancel}
We also propose an estimate-and-cancel scheme which is even simpler than DC scheme.
In estimate-and-cancel, M1 sends a scaled version of the waveform of the interfering signal over the side-channel.
At the receiver M2, we first estimate the interfering signal from the side-channel received signal $Y_3$ , then
rescale it according to $X_2$ and cancel it out from $Y_1$.
Finally decode the signal-of-interest at the main-channel receiver at M2 using single-user decoder.
Note that in EC scheme, we have less flexibility in using the side-channel since the bandwidth of the side-channel is required to be
equal or larger than that of the main-channel.
Let $X_3=K X_2$, where $K$ is a scalar~($K\geq0$) satisfying the per-node power constraint, i.e., $\mathbb{E}(|X_2+X_3|^2)\leq P_2$.
Thus $\mathbb{E}(|X_2|^2)\leq \frac{P_2}{(1+K)^2}$.

We adopt independently Gaussian codebooks for $X_1$ and $X_2$ with $X_1\sim\mathcal{N}(0,P_1), X_2\sim\mathcal{N}\left(0,\frac{P_2}{(1+K)^2}\right)$.
Side-channel signal $X_3$ is correlated with $X_2$ such that $X_{3}\sim\mathcal{N}\left(0,\frac{K^2P_2}{(1+K)^2}\right)$, where $K\geq 0$.
The resulting rate of EC is given as follows, with calculation of rate shown in Appendix~\ref{Append2}.
\begin{Proposition} \label {ECrates}For Gaussian side-channel assisted three-node full-duplex network, the following rate region is achievable for estimate-and-cancel
\begin{gather}
\begin{aligned}
R_1\leq&C\left(\frac{\mathrm {SNR_1}\left(1+\frac{K^2 \rm SNR_{side}}{(1+K)^2}\right)}{1+\frac{K^2 \mathrm{SNR_{side}}}{(1+K)^2}+\frac{\rm INR}{(1+K)^2}}\right)\\
R_2\leq&C\left(\frac{\rm SNR_2}{(1+K)^2}\right),
\end{aligned}
\end{gather}
where $K\geq0$.
\end{Proposition}
\subsection{New Outer Bound \label{sec:outer}}
In this section, we derive a new outer bound to characterize the performance of the four schemes discussed in Sections~\ref{BCscheme} and~\ref{sec:simpler}.
When there is no INI in the side-channel assisted three-node full-duplex network, we can obtain the no-interference outer bound,
\begin{gather}
\begin{aligned}
  R_{1}&\leq C_1\\
    R_{2}&\leq C_2.\label{Noouter}
\end{aligned}
\end{gather}

However, the no-interference outer bound can be arbitrarily loose.
We notice that one of the transmitter and receiver pairs in Figure~\ref{fig.10} is co-located at the base station, thus causal information of M1 is available
at the transmitter BS.
Hence we derive a new outer bound on the capacity region of Gaussian side-channel assisted three-node network as follows.
\begin{theorem}\label{sideub}
For Gaussian side-channel assisted three-node full-duplex network under per-node power constraint, the capacity region is outer bounded by
\begin{equation}
\begin{aligned}
  R_{1}\leq& C_1\\
    R_2\leq& C\left(\bar{\lambda}\rm SNR_2\right) \nonumber
\end{aligned}
\end{equation}
\begin{equation}
\begin{aligned}
R_1+R_2\leq & C\left(\frac{\bar{\lambda}\mathrm{SNR_2}}{1+\bar{\lambda}\mathrm{INR}}\right)+WC\left(\frac{\lambda\mathrm{SNR_{side}}}{W}\right)\\
&+C\left(\mathrm{SNR_1}+\bar{\lambda}\mathrm{INR}+2\sqrt{\bar{\lambda}\mathrm{SNR_1 INR}}\right), \label{outerstrong}
\end{aligned}
\end{equation}
where $\lambda\in[0,1], \lambda+\bar{\lambda}=1$.
\end{theorem}
\begin{proof}
The first two bounds on $R_1$ and $R_2$ come from the point-to-point capacity of AWGN channel.
The third bound on $R_1+R_2$ is the sum-capacity upper bound. Messages $S_1$ and $S_2$ are chosen uniformly and independently at random,
and $X_{1i}$ is a function of $(S_1,Y_2^{i-1})$ for $i\in[1,n]$.  We define a genie $V_2=\gamma^m_{21}X_2+Z_1$. The sum-capacity upper bound is derived by providing the genie $V^n_{\mathrm{2}}$ to the base station. 
From Fano's inequality, we have $H(S_i|Y_i{}^n)\leq \epsilon_{n,i}$ for $i=1,2$, where $n$ goes
to infinity as $\epsilon_{n,i}$ goes to zero. Let $\epsilon_{n}=\epsilon_{n,1}+\epsilon_{n,2}$. Thus for any codebook of block length $n$, we have
\begin{eqnarray}
&&n(R_1+R_2-\epsilon_n) \\
&\leq &I(S_1;Y_1^n,Y_3^n)+I(S_2;Y_2^n|S_1) \\
&=&I(S_1;Y_1^n)+I(S_1;Y_3^n|Y_1^n)+I(S_2;Y_2{}^n|S_1)\\
&=&I(S_1;Y_1^n)+H(Y_3^n|Y_1^n)-H(Y_3^n|Y_1^n,S_1)\nonumber\\
&&+I(S_2;Y_2{}^n|S_1) \label{aa}\\
&\leq&I(S_1;Y_1^n)+I(X_3^n;Y_3^n) +I(S_2;Y_2{}^n|S_1)\label{bb}\\
&=&h(Y_1^n)+h(Y_3^n)-h(Z_3^n)\nonumber\\
&&+\underbrace{h(Y_2^n|S_1)-h(Y_2^n|S_2,S_1)-h(Y_1^n|S_1)}_{U}\label{cc}
\end{eqnarray}
where (\ref{bb}) follows from $H(Y_3^n|Y_1^n)\leq H(Y_3^n)$ and $-H(Y_3^n|Y_1^n,S_1)\leq -H(Y_3^n|X_3^n)$.
$U$ in (\ref{cc}) can be further upper bounded as
\begin{eqnarray}
U&\leq& \sum_{i=1}^n h(Y_{2i}|X_{1i},V_{2i})-\sum_{i=1}^n\big(h(Z_{1i})+h(Z_{2i})\big) \label{d}\\
&\leq&  \sum_{i=1}^n h(Y_{2i}|V_{2i})-\sum_{i=1}^n\big(h(Z_{1i})+h(Z_{2i})\big) \label{e}
\end{eqnarray}
where (\ref{d}) follows from result in \cite{achal} (in Appendix~C); (\ref{e}) follows because removing
condition does not reduce entropy.

Combining the results above and applying the chain rule, we have
\begin{gather}
\begin{aligned}
R_1+R_2-\epsilon_n&\leq \frac{1}{n}\sum_{i=1}^n\bigg(h(Y_{1i})+h(Y_{3i})+h(Y_{2i}|V_{2i})-\big[h(Z_{1i})\\
&+h(Z_{2i})+h(Z_{3i})\big]\bigg).\nonumber
\end{aligned}
\end{gather}
Since Gaussian inputs maximize entropy, therefore the sum-rate is upper bounded when $X_1\sim \mathbf{N(0,P_1)}$,  $X_2\sim \mathbf{N(0,\bar{\lambda}P_2)}$ and $X_3\sim \mathbf{N(0,\lambda P_2)}$, where $\lambda\in[0,1], \bar{\lambda}+\lambda=1$. The correlation coefficient between the codewords $X_1$ and $X_2$ is $\rho$, where $\rho\in[-1,1].$ 
Hence the sum-capacity upper bound can be expressed as 
\begin{equation}
\begin{aligned}
R_1+R_2 &\leq \sup_{\left|\rho\right|\in[0,1]} C\left(\frac{\bar{\lambda}\mathrm {SNR_2}}{1+\bar{\lambda}\mathrm{INR}}\right)+WC\left(\frac{\lambda\mathrm{SNR_{side}}}{W}\right)\nonumber\\
&+C\left(\mathrm{SNR_1}+\bar{\lambda}\mathrm{INR}+2\left|\rho\right|\sqrt{\bar{\lambda}\mathrm{SNR_1 INR}}\right).\\
&\leq C\left(\frac{\bar{\lambda}\mathrm{SNR_2}}{1+\bar{\lambda}\mathrm{INR}}\right)+WC\left(\frac{\lambda\mathrm{SNR_{side}}}{W}\right)\\
&+C\left(\mathrm{SNR_1}+\bar{\lambda}\mathrm{INR}+2\sqrt{\bar{\lambda}\mathrm{SNR_1 INR}}\right).
\end{aligned}
\end{equation}
\end{proof}

\subsection{Within One Bit of the Capacity Region \label{sec:one-bit}}
In this section, we show that BC is within one-bit of the new outer bound derived in Section~\ref{sec:outer}. Furthermore, we also characterize the capacity gap of DC and EC schemes, and show that they are also approximately optimal in specific regimes.

Let $\mathcal{R_{\rm BC}}=\bigcup_{\beta,\lambda}\mathfrak{R_{\mathrm{BC}}}(\beta,\lambda)$ denote the achievable rate region for bin-and-cancel including all possible power split at Node~M1, where
$\mathfrak{R_{\mathrm{BC}}}(\beta,\lambda)$ denote the achievable rate region for a fixed power split between common and private information of M1, as well as
side-channel and main-channel at M1.
The main result is given in the following theorem. 
\begin{theorem} \label{onebitbc}
The achievable region $\mathcal{R_{\rm BC}}$ is within 1 bit/s/Hz of the capacity region of Gaussian side-channel assisted three-node network, for all values of channel
parameters and bandwidth ratio.
\end{theorem}
\begin{proof}
Let $\delta_{R_1}^{\rm BC}$ denote the difference between the upper bound on $R_1$ and achievable rate $R_1$ in $\mathcal{R_{\rm BC}}$. Likewise, we have $\delta_{R_2}^{\rm BC}$ and
$\delta_{R_1+R_2}^{\rm BC}$, where all rates are divided by the total bandwidth $W_m+W_s$. For a given $\lambda$, i.e., fixing the power allocated to the side channel, we show that there exist an achievable scheme which achieves
the corresponding outer bound to within 1 bit/s/Hz. In order to prove the rate pair $(R_1-1\times (W_m+W_s), R_2-1\times (W_m+W_s))$ in $\mathcal{R_{\rm BC}}$ is achievable for any $(R_1, R_2)$ in the capacity region of Gaussian three-node full-duplex with side-channel, we just need to show that
\begin{gather}
\begin{aligned}
\delta_{R_1}^{\rm BC}<&1\\
\delta_{R_2}^{\rm BC}<&1\\
\delta_{R_1+R_2}^{\rm BC}<&2.\label{one-bit}
\end{aligned}
\end{gather}
Since the inner bound in the weak interference regime is different from that in the strong interference regime.
We will prove (\ref{one-bit}) for two different cases.
\begin{enumerate}
\item In the weak interference regime where $\rm INR\leq SNR_2$,
we use the simplified BC scheme for comparison when the power splitting of private message
of M1 is set at the noise level at the receiver M2, i.e., $\mathrm{INR}_p=
\beta\bar{\lambda}\mathrm{INR}=1$, when $\mathrm{INR}\geq 1$.\footnote{If $\mathrm{INR}<1$, the interference is even weaker than the Gaussian noise, one can easily show that treating interference as noise can achieve the capacity region to within one bit.}
Thus $\mathfrak{R_{\mathrm{BC}}}(\beta=\frac{1}{\bar{\lambda}\rm INR})\subset\mathcal{R_{\rm BC}}$.
We can compute $\mathfrak{R_{\mathrm{BC}}}(\beta=\frac{1}{\bar{\lambda}\rm INR})$ directly from Proposition~\ref{BCrate}, by
comparing with the new outer bound in (\ref{outerstrong}), it is straightforward to show that
\begin{gather}
\begin{aligned}
&\!\!\delta_{R_1}^{\rm BC}\leq \frac{W_m}{W_m+W_s}\mathrm{log}\left(1+\frac{\rm SNR_1}{2+\rm SNR_1}\right)\\
&\!\!~~~~~<\frac{1}{1+W}\leq 1\\
&\!\!\delta_{R_2}^{\rm BC}\leq\frac{W_m}{W_m\!+\!W_s}\!\max\left\{\!0,1-W\mathrm{log}\left(1+\frac{\lambda \rm SNR_{side}}{W}\right)\right.\\
&\!\!~~~~~\left.+\mathrm{log}\left(\!\frac{1+\bar{\lambda}\rm SNR_2}{1\!+\!\bar{\lambda}\mathrm{INR}}\frac{\rm INR}{\rm INR+SNR_2}\!\right)\right\}\\
&\!\!~~~~~~<\frac{1}{1+W}\leq1\\
&\!\!\delta_{R_1+R_2}^{\rm BC}\leq \frac{W_m}{W_m+W_s}\bigg(2\\
&+\mathrm{log}\left(1-\frac{\rm SNR_2}{\mathrm{INR}+\mathrm{SNR_2}+\bar{\lambda}\mathrm {INR SNR_2}+\bar{\lambda}\rm INR^2 }\right)\bigg)\\
&\!\!~~~~~~~~~<\frac{2}{1+W}\leq 2.
\end{aligned}
\end{gather}

\item In the strong interference regime where $\rm INR\geq SNR_2$, by substituting the outer bound in (\ref{outerstrong}) and inner bound achieved by BC in (\ref{BCstrong}), we have
\begin{gather}
\begin{aligned}
\delta_{R_1}^{\rm BC}&=0\\
\delta_{R_2}^{\rm BC}&=0\\
\delta_{R_1+R_2}^{\rm BC}&\leq\frac{W_m}{W_m+W_s}\mathrm{log}\left(1+\frac{2\sqrt{\bar{\lambda}\rm SNR_1 INR}}{1+\mathrm{ SNR_1}+\bar{\lambda}\mathrm{INR}}\right)\\
&\leq \frac{1}{1+W}\leq1.
\end{aligned}
\end{gather}
\end{enumerate}
\end{proof}
We also characterize the capacity gap for two simpler schemes, DC and EC, and the following theorem quantifies the conditions
where DC and EC can achieve within half bit of the capacity region, respectively.
\begin{theorem}
\begin{enumerate}
\item For decode-and-cancel, the achievable region $\mathcal{R_{\rm DC}}$ is within $\frac{1}{2}$~bits/s/Hz of the capacity region of Gaussian side-channel assisted
three-node network when $W\geq1$ and $\rm SNR_{side}\geq SNR_2$.
\item For estimate-and-cancel, the achievable region $\mathcal{R_{\rm EC}}$ is within $\frac{1}{2}$~bits/s/Hz of the capacity region of Gaussian
side-channel assisted three-node network when $W\in\mathbb{N}^+$ and $\rm SNR_{side}\geq \left(1+\frac{2}{\sqrt{2}-1}\right)(INR-2)$.
\end{enumerate}
\end{theorem}
\begin{proof}
\begin{enumerate}
\item We first prove the condition on half bit capacity gap for DC scheme. It is not difficult to find that the achievable rate of DC is an increasing function of $W$.
When $W\geq1$, according to \cite{Jingwen}, $R_2\geq C\left(\frac{\rm SNR_2 SNR_{side}}{\rm SNR_2+SNR_{side}}\right)$. Hence when $W\geq1$ and $\rm SNR_{side}\geq SNR_2$, comparing with the no-interference outer bound in (\ref{Noouter}) we prove that

\vspace{-10pt}
{\small
\begin{gather}
\begin{aligned}
&\delta_{R_1}^{\rm DC}=0\\
&\delta_{R_2}^{\rm DC}\leq \frac{W_m}{W_m+W_s}\mathrm{log}\left(1+\mathrm{SNR_2}\right)\\
&- \frac{W_m}{W_m+W_s}\mathrm{log}\left(\frac{\rm SNR_2 SNR_{side}}{\rm SNR_2+SNR_{side}}\right)\\
&= \frac{W_m}{W_m\!+\!W_s}\mathrm{log}\left(\!1\!+\!\frac{\rm SNR_2^2}{\rm SNR_2\!+\!SNR_{side}\!+\!SNR_2 SNR_{side}}\!\right)\\
&\leq \frac{1}{1+W}\leq\frac{1}{2}
\end{aligned}
\end{gather}
}%
\item For EC scheme, the bandwidth ratio between side-channel and main-channel is required to be greater than 1, i.e., $W\geq1$.
Comparing the rate region of EC in Propostion~\ref{ECrates} with the no-interference outer bound we have

\vspace{-10pt}
{\small
\begin{gather}
\begin{aligned}
\delta_{R_1}^{\rm EC}&\leq\frac{W_m}{W_m+W_s}\mathrm{log}\left(1+\mathrm{SNR_1}\right)\\
&-\frac{W_m}{W_m+W_s}\mathrm{log}\left(1+\frac{\mathrm {SNR_1}\left(1+\frac{K^2 \rm SNR_{side}}{(1+K)^2}\right)}{1+\frac{K^2 \mathrm{SNR_{side}}}{(1+K)^2}+\frac{\rm INR}{(1+K)^2}}\right)\\
&=\frac{W_m}{W_m+W_s}\mathrm{log}\left(1+\right.\\
&\left.\frac{\frac{\rm SNR_1INR}{(1+K)^2}}{1+\frac{K^2\mathrm{SNR_{side}}}{(1+K)^2}+\frac{\rm INR}{(1+K)^2}+\mathrm{SNR_1}\left(1+\frac{K^2\mathrm{SNR_{side}}}{(1+K)^2}\right)}\right)\\
\delta_{R_2}^{\rm EC}&\leq\frac{W_m}{W_m+W_s}\mathrm{log}\left(1+\mathrm{SNR_2}\right)\\
&-\frac{W_m}{W_m+W_s}\mathrm{log}\left(1+\frac{\rm SNR_2}{(1+K)^2}\right)\\
&=\frac{W_m}{W_m+W_s}\mathrm{log}\left(1+\frac{\frac{K(K+2)\mathrm{SNR_2}}{(1+K)^2}}{1+\frac{\rm SNR_2}{(1+K)^2}}\right).
\end{aligned}
\end{gather}
}%
The achievable rate region of EC is the union of rate pairs for all possible values of $K$, where $K\geq 0$.
Therefore, when $\rm SNR_{side}\geq \left(1+\frac{2}{\sqrt{2}-1}\right)(INR-2)$,
and we chose $K=\sqrt{2}-1$,
the following inequalities will hold simultaneously
 \begin{gather}
\begin{aligned}
\delta_{R_1}^{\rm EC}& \leq \frac{1}{1+W}\leq\frac{1}{2}\\
\delta_{R_2}^{\rm EC}&\leq \frac{1}{1+W}\leq\frac{1}{2}.
\end{aligned}
\end{gather}
\end{enumerate}
\end{proof}
\subsection{Discussion of capacity analysis results}
Practical communication systems usually operate in the inference-limited regime, so the strength of the signal and interference are much larger
than the thermal noise. Our one-bit capacity gap result shows that for all channel parameters and bandwidth ratio, BC scheme is asymptotic optimal in high SNR, INR limit, because one bit is relatively small compared to the rate of the users. For the two simple schemes, DC and EC can
achieve the near-optimal (half-bit gap) performance only in certain
regimes of channel values. From the capacity analysis in previous section, we also notice that the capacity gap for each scheme is inversely proportional to the
bandwidth ratio $W$, where $W\in[0, W_{\max}]$.
Often, in practical wireless communication, the bandwidth of side-channel is different from the main-channel. For example, ISM band radio can
operate on a bandwidth up to 80~MHz in 2.4~GHz. In contrast, generally the bandwidth of current 3G/4G radio is 20~MHz.
Motivated by
the application of leveraging ISM band as the side-channel
in a cellular network, the bandwidth of side-channel
usually is much larger than that of main-channel. If $W=W_{\max}$,
for all values of SNR and INR, BC is near-optimal which achieves within $\frac{1}{1+W_{\max}}$ bits/s/Hz of the capacity region.

\subsection{No side-channel Special Case}

We can obtain the system model without side-channel when bandwidth ratio $W=0$ and $Y_3=0$ in Equation~(\ref{systemequ}), which
can be modeled as causal cognitive Z-interference channel~\cite{Cog1}.
In the cognitive Z-interference channel, we assume that at time $i$, the encoder of base station will have access to the information sequence sent by M1 up to time $i-1$.

When we do not use the causal information at BS, the resulting channel is classic Z-channel. The rate region of classic Z-channel is the capacity inner bound of cognitive Z-interference
channel. And the rate region of classic Z-channel is a special case of the achievable rate region of BC scheme in Proposition~\ref{BCrate} when $W=0$.
The sum-rate can be achieved by treating interference as noise in the weak regime
and joint decoding in the strong interference regime.

The state-of-art approaches based on block Markov codes \cite{Cog1, Cog2} do not seem to enhance the achievable rate region over Z-channel~(to the best of our knowledge), a counter-intuitive result which may be due to the causality of information at BS. At time $i$, the way that the transmitter BS can obtain the causal message of M1 is by decoding the message through
the direct link between M1 and BS before time $i$, so we can not take advantage of the decoding schemes for block Markov codes which require delay.
Thus we conjecture
that this causal information will not help improve the achievable sum-rate over classic Z-channel.
Nevertheless, we can show that the sum-capacity is established in the high $\rm SNR, INR$ limit.
As $\rm SNR_1, SNR_2, INR$ $\rightarrow\infty$, with their ratios being kept constant, we obtain the following theorem.
\begin{theorem}
For the Gaussian three-node full-duplex network, the asymptotic sum-capacity is established
  \begin{gather}
  \begin{aligned}
 & C_{\mathrm{sum}}^{\rm No-SC}  =\\
 & \begin{cases}
  C(\rm {SNR}_2) + C\left(\frac{\rm {SNR}_1}{1+\rm {INR}}\right)~\rm INR\leq \rm SNR_2, \\
   C\left(\rm {SNR}_1 + \rm {INR}\right)~\rm SNR_2\leq\rm INR\leq\rm SNR_2(1+\rm SNR_1),\\
   C(\rm {SNR}_1) + C(\rm {SNR}_2)~\rm INR\geq \rm SNR_2(1+\rm SNR_1).
   \end{cases}
   \end{aligned}
\end{gather}
\end{theorem}
\begin{figure*}[!t]
\normalsize
\setcounter{MYtempeqncnt}{\value{equation}}
\setcounter{equation}{35}
\begin{align}
\label{pp1}
 \lim_{\substack{\rm SNR_1,SNR_2,\\
\rm INR \to \infty}} \!\frac{\overline{C}^{\rm No-SC}_{\mathrm{sum}}}{\underline{C}^{\rm No-SC}_{\mathrm{sum}}}&=\lim_{\substack{\rm SNR_1,SNR_2,\\
\rm INR \to \infty}}\!\frac{\mathrm{log(1+\alpha SNR_2)\!+\!log\left(\frac{1+INR+SNR_1+2\sqrt{\bar{\alpha}\mathrm{SNR_1INR}}}{1+\alpha\mathrm{INR}}\right)}}{\mathrm{log}(1+\beta \mathrm{SNR}_2)\!+\!\mathrm{log}\left(\frac{1+\mathrm{INR+SNR_1}}{1+\beta\mathrm{INR}}\right)} \nonumber \\
&=\lim_{\substack{\rm SNR_1,SNR_2,\\
\rm INR \to \infty}} \!\! \frac{\mathrm{log}(\alpha\mathrm{SNR_2})\!+\!\mathrm{log}(\mathrm{INR\!+\!SNR_1}\!+\!2\sqrt{\bar{\alpha}\mathrm{SNR_1INR}})\!-\!\mathrm{log}(\alpha\mathrm{INR})}{\mathrm{log}(\beta\mathrm{SNR_2})\!+\!\mathrm{log(SNR_1\!+\!INR)}\!-\!\mathrm{log}(\beta\mathrm{INR})}\nonumber \\
&=\lim_{\substack{\rm SNR_1,SNR_2,\\
\rm INR \to \infty}} \!\! \frac{\mathrm{log}(\mathrm{SNR_2})+\mathrm{log}(\mathrm{INR+SNR_1})-\mathrm{log}(\mathrm{INR})+f_1(\alpha)}{\mathrm{log}(\mathrm{SNR_2})+\mathrm{log(SNR_1+INR)}-\mathrm{log}(\mathrm{INR})+f_2(\beta)}=1,
\end{align}
where $f_1(\alpha)$ and $f_2(\beta)$ are constant for given $\alpha$ and $\beta$, respectively.\\
\text{} 
\setcounter{equation}{\value{MYtempeqncnt}}
\text{}
\hrulefill
\vspace*{8pt}
\end{figure*}
\begin{figure*}[t] 
\vspace*{5pt}
\normalsize
\setcounter{MYtempeqncnt}{\value{equation}}
\setcounter{equation}{36}
\begin{gather}
\begin{aligned}
\label{pp2}
&\lim_{\substack{\rm SNR_1,SNR_2,\\
\rm INR \to \infty}} \frac{\rm INR^*}{\rm SNR_2(1+\rm SNR_1)}=\lim_{\substack{\rm SNR_1, SNR_2,\\
\rm INR \to \infty}}  \frac{2\rm SNR_1+SNR_2(1+SNR_1)-2\sqrt{SNR_1(SNR_1+SNR_2+SNR_1SNR_2)}}{\rm SNR_2(1+SNR_1)}=1.
\end{aligned}
\end{gather}
\setcounter{equation}{\value{MYtempeqncnt}}
\hrulefill
\end{figure*}
\begin{proof}
We only need to prove that the upper bound on the sum-capacity of Gaussian three-node full-duplex network is tight as $\rm SNR_1, SNR_2, INR$ $\rightarrow\infty$.
We can compute the sum-capacity upper bound directly from Theorem~\ref{sideub} and lower bound from Proposition~\ref{BCrate} by substituting $W=0$ and $\lambda=0$.
We use $\overline{C}^{\rm No-SC}_{\mathrm{sum}}$ and $\underline{C}^{\rm No-SC}_{\mathrm{sum}}$ to denote upper bound and lower bound on the no-side-channel sum-capacity, respectively.

In the weak interference regime where
$\rm INR\leq \rm SNR_2$, when $\rm SNR_1, SNR_2, INR$ $\rightarrow\infty$, the lower bound matches with upper bound as shown in (\ref{pp1}) at the top of the page.

The strong interference regime for the sum-capacity lower bound is defined as $\rm SNR_2\leq\rm INR\leq\rm SNR_2(1+\rm SNR_1)$, and for the sum-capacity upper bound is defined as $\rm SNR_2\leq\rm INR\leq\rm INR^*$, where
$\rm INR^*=2\rm SNR_1+SNR_2(1+SNR_1)
-2\sqrt{SNR_1(SNR_1+SNR_2+SNR_1SNR_2)}.$
Thus in the strong interference regime, using the result obtained in (\ref{pp2}) at the top of the page, we have
\begin{gather}
\begin{aligned}
&\lim_{\mathrm{SNR_1, SNR_2, INR} \to \infty}\frac{\overline{C}^{\rm No-SC}_{\mathrm{sum}}}{\underline{C}^{\rm No-SC}_{\mathrm{sum}}}= \\ &\lim_{\mathrm{SNR_1, SNR_2, INR} \to \infty}\frac{C\left(\rm SNR_1+INR+2\sqrt{SNR_1INR}\right)}{  C\left(\rm {SNR}_1+\rm {INR}\right)} \nonumber\\
&~~~~~~~~~=1.
\end{aligned}
\end{gather}

If $\rm INR\geq INR^*$  or $\rm INR\geq \rm SNR_2(1+\rm SNR_1)$, the constraint on $R_1+R_2$ in (\ref{outerstrong}) becomes redundant and the sum-capacity upper bound reduces to (\ref{PTP}) when $W=0, \lambda=0$.
\end{proof}

Therefore, in the high SNR and INR limit, for the Gaussian three-node full-duplex network, treating interference as noise and joint decoding are
asymptotically sum-capacity achieving in weak and strong interference regimes, respectively.
\section{Asymptotic and Finite SNR Multiplexing Gain Comparisons}
\subsection{High SNR Multiplexing Gain }
For simplicity, we assume $\rm SNR_1=SNR_2=SNR$ and we define $\mu=\frac{\rm log~INR}{\rm log~SNR}$, $\nu=\frac{\rm log~SNR_{\rm side}}{\rm log~SNR}$. Parameter $\mu$ captures the interference level, parameters $\nu$ and $\frac{\nu}{\mu}$ capture the side-channel level compared with main-channel.

In this paper, we use multiplexing gain as the metric to characterize the rate improvement by leveraging the side-channel which is given as
\setcounter{equation}{37}
\begin{gather}
\begin{aligned}
M=\frac{R_{\mathrm{sum}}(\rm SNR,INR)}{C_{\rm single-user}(\rm SNR)} ,
\end{aligned}
\end{gather}
where $C_{\rm single-user}=W_m\rm log\left(1+\rm SNR\right).$

The single-user capacity is the maximum uplink/downlink rate of the three-node network in the main-channel.
Ideally, a perfect full-duplex can achieve the maximum multiplexing gain of 2 when there is no interference, which corresponds to the largest sum-capacity of $2W_m \log\left(1+\rm SNR\right)$.
Assuming $\rm SNR\gg1, INR\gg1$ and $\mathrm {SNR_{side}}\gg W$, we can compute the asymptotic multiplexing gain of our proposed schemes as $\rm SNR, SNR_{\rm side}, INR$ $\rightarrow\infty$ and $\mu\geq0, \nu\geq0, W\in[0,W_{\rm max}]$.

First we derive the asymptotic multiplexing gain of the no-side-channel sum-capacity over single-user capacity in the high SNR and INR limit, which is given by
\begin{gather}
\begin{aligned}
  M^{\rm No-SC}=\frac{C_{\mathrm{sum}}^{\rm No-SC}}{C_{\rm single-user}}=
  \begin{cases}
  2-\mu &  ~~0\leq \mu<1,\\
   \mu &~~ 1\leq\mu<2,\\
   2     &~~ \mu\geq2.
   \end{cases}\label{No-mg}
\end{aligned}
\end{gather}

When $\mu=0$, there is no interference. And $0\leq \mu<1$ corresponds to the weak interference regime while $\mu\geq1$ corresponds to the strong and very strong interference regimes. In the no-side-channel case, only when interference level is zero or above very strong can a maximum value of 2 for asymptotic multiplexing gain be achieved. However, INI causes the multiplexing gain to reduce significantly in all other regimes.
\begin{figure*}[t] 
\vspace*{4pt}
\normalsize
\setcounter{MYtempeqncnt}{\value{equation}}
\setcounter{equation}{40}
\begin{gather}
\begin{aligned}
\label{pp3}
\frac{R_{\mathrm{sum}}^{\rm BC}}{C_{\rm single-user}}&=\frac{W_m\min\left\{\rm log\left(1+\frac{SNR}{2}\right)+\rm log(1+\bar{\lambda}SNR),\rm log\left(1+\frac{SNR+\bar{\lambda}INR-1}{2}\right)+\rm log\left(1+\frac{SNR}{INR}\right)+W\rm log\left(1+\frac{\lambda SNR_{side}}{W}\right)\!\right\}}{W_m\rm log(1+SNR)}\\
&\approx\frac{\min\left\{\rm 2log\left(SNR\right),\rm 2log\left(SNR\right)-{\mathrm log(INR)}+W\rm log\left(SNR_{side}\right)\!\right\}}{\rm log(SNR)}\\
&\approx\min\{2,2+W\nu-\mu\}.
\end{aligned}
\end{gather}
\setcounter{equation}{\value{MYtempeqncnt}}
\hrulefill
\end{figure*}
For BC, the achievable sum-rate can be expressed as
\begin{gather}
\begin{aligned}
R_{\mathrm{sum}}^{\rm BC}\leq &\max_{\substack{0\leq\lambda\leq1\\ 0\leq\beta\leq 1}}
C\left(\frac{\mathrm{SNR_1}}{\beta \bar{\lambda}\mathrm{INR}+1}\right)+\min\!\left\lbrace C\left(\bar{\lambda}\mathrm{SNR_2}\right),\right.\\
&\left.C\left(\beta\bar{\lambda}\mathrm{SNR_2}\right)+C\left(\frac{\bar{\beta}\bar{\lambda}\mathrm{INR}}{\mathrm{SNR_1}+\beta\bar{\lambda}\mathrm{INR}+1}\right)\right.\\
&\left.+WC\left(\frac{\lambda\mathrm{SNR_{side}}}{W}\right)\!\right\rbrace.
\label{BCsum}
\end{aligned}
\end{gather}
And $\beta^*$ which maximizes achievable sum-rate of BC in~(\ref{BCsum}) is
\begin{eqnarray}
 \beta^*(a)=
  \begin{cases}
    \frac{(1+\bar{\lambda}\mathrm{SNR_2})(1+\mathrm{SNR_1})-a(1+\mathrm{SNR_1}+\bar{\lambda}\mathrm{INR})}{a\bar{\lambda}\mathrm{SNR_2}(1+\mathrm{SNR_1}+\bar{\lambda}\mathrm{INR})-\bar{\lambda}\mathrm{INR}(1+\bar{\lambda}\mathrm{SNR_2})}
    & \mu\leq1, \\
    0      & \mu\geq1,
   \end{cases}\nonumber
\end{eqnarray}
where $a=\left(1+\frac{\lambda\mathrm{SNR_{side}}}{W}\right)^W$.

As $\rm SNR, SNR_{\rm side}, INR$ $\rightarrow\infty$, in the weak interference regime~(i.e., $0\leq\mu<1$),
set the power for the private message
of M1 at the level of Gaussian noise at the receiver M2, namely, let $\beta=\frac{1}{\bar{\lambda}\rm INR}$. Thus for a fixed $\lambda~(\lambda\neq 0)$, we derive the expression in (\ref{pp3}) at the top of the page.

In the strong interference regime where $\mu\geq1$, $\beta=0$. As $\rm SNR, SNR_{\rm side}, INR$ $\rightarrow\infty$, we have
\setcounter{equation}{41}
\begin{gather}
\begin{aligned}
&\frac{R_{\mathrm{sum}}^{\rm BC}}{C_{\rm single-user}}\approx\\
&\frac{\min\left\{\rm 2log\left(SNR\right),\rm log (INR)+W\rm log\left(SNR_{side}\right)\!\right\}}{\rm log(SNR)}\\
&\approx\min\{2,\mu+W\nu\}.
\end{aligned}
\end{gather}

From Theorem~\ref{onebitbc}, we have proved that the inner bound achieved by BC is within one bit of the outer bound by a simple common-private power split when the power splitting of private message of M1 is set at the level of Gaussian noise at the receiver M2. Hence in the high SNR and INR limit, BC scheme is asymptotic optimal in that it is capacity-achieving for the Gaussian side-channel assisted three-node network.
Therefore, the asymptotic multiplexing gain of the sum-capacity of the Gaussian side-channel assisted three-node full-duplex network over single-user capacity is
\begin{gather}
\begin{aligned}
 & M^{\rm side-channel}=\frac{C_{\mathrm{sum}}^{\rm side-channel}}{C_{\rm single-user}}\\
 & =M^{\rm BC}=
  \begin{cases}
  \min\left\{2,2+W\nu-\mu\right\}~0\leq \mu<1,\\
   \min\{2,\mu+W\nu\}~\mu\geq1.
   \end{cases}\label{BC-mg}
\end{aligned}
\end{gather}
By comparing (\ref{BC-mg}) to (\ref{No-mg}), we can derive the multiplexing gain improvement of the side-channel sum-capacity over the no-side-channel sum-capacity,
\begin{gather}
\begin{aligned}
 \frac{M^{\rm side-channel}}{M^{\rm No-SC}}&\!=\!
 \begin{cases}
  \min\left\{\frac{2}{2-\mu},1+\frac{W\nu}{2-\mu}\right\}~0\leq \mu<1,\\
   \min\{\frac{2}{\mu},1+\frac{W\nu}{\mu}\}~1\leq\mu<2.
   \end{cases}\label{side-no-mg}
\end{aligned}
\end{gather}
When $W\nu=0$, namely, there is no side-channel available, (\ref{BC-mg}) will degenerate into (\ref{No-mg}).
From~(\ref{side-no-mg}), we can see that the improvement is dominated by the interference level~(i.e., $\mu$), as well as the side-channel condition~(i.e., $W\nu$). In the weak interference regime where $\mu\in[0,1)$, the improvement is an increasing function of the interference level,
while in the strong interference regime where $\mu\geq1$, the improvement will decrease as $\mu$ increases. Therefore, we can obtain the maximum multiplexing gain improvement when $\mu=1$. The improvement rises as $W\nu$ increases till reaching the maximum multiplexing gain of 2. If $W\nu\geq1$ and $\mu=1$, then leveraging the side-channel can double the asymptotic multiplexing gain of the no-side-channel case.

Similarly, we can derive the asymptotic multiplexing of the suboptimal schemes,
\begin{gather}
\begin{aligned}
  M^{\rm CC}=\frac{R_{\mathrm{sum}}^{\rm CC}}{C_{\rm single-user}}=
  \begin{cases}
   \min\left\{2,2+W\nu-\mu\right\}~0\leq \mu<1,\\
   \min\left\{2,1+W\nu\right\}~\mu\geq1.
   \end{cases}\label{CC-mg}
\end{aligned}
\end{gather}
\begin{gather}
\begin{aligned}
  M^{\rm DC}=\frac{R_{\mathrm{sum}}^{\rm DC}}{C_{\rm single-user}}=\min\{2,1+W\nu\} ~~\mu\geq0, \label{DC-mg}
\end{aligned}
\end{gather}
\begin{gather}
\begin{aligned}
&  M^{\rm EC}=\frac{R_{\mathrm{sum}}^{\rm EC}}{C_{\rm single-user}}=\\
&\left\{
  \begin{array}{l l}
  \min\{2,2+\nu-\mu\}~~~0\leq \mu<\nu+1,W\in\mathbb{N}^+,\\
   1~~~\mu\geq\nu+1,W\in\mathbb{N}^+.
    \end{array} \right.\label{EC-mg}
\end{aligned}
\end{gather}

Comparing (\ref{BC-mg}) and (\ref{CC-mg})-(\ref{EC-mg}), we can find out that for all parameters, CC scheme is asymptotic sum-capacity achieving in the weak interference regime, while in the strong interference regime, the performance of CC coincides with that of the DC scheme. Moreover, it is straightforward
to show that CC achieves larger asymptotic multiplexing gain than DC and EC in the weak interference regime.

Next we discuss the effect of the side-channel on the multiplexing gain of our proposed schemes.  When $W=1$, EC is also asymptotically sum-capacity achieving in the weak interference regime in addition to BC and CC. Specifically, for EC scheme, when $W\in\mathbb{N}^+$ and $\nu\geq\mu$, EC can achieve a asymptotic multiplexing gain of 2 for all interference levels. In addition, DC can achieve a asymptotic multiplexing gain of 2 for all interference levels when $W\nu\geq1$. For BC and CC, a asymptotic multiplexing gain of 2 can be achieved for all interference levels when $W\nu\geq\mu$ for $\mu\in[0,1)$, and $W\nu\geq1$ for $\mu\geq1$.
\begin{figure}[t!] 
\begin{center}
  \subfigure[$W<1$ and $\nu\geq\mu$]
  {\label{asmg1}{\includegraphics[width=0.55\textwidth,trim = 30mm
      34mm 20mm 40mm, clip]{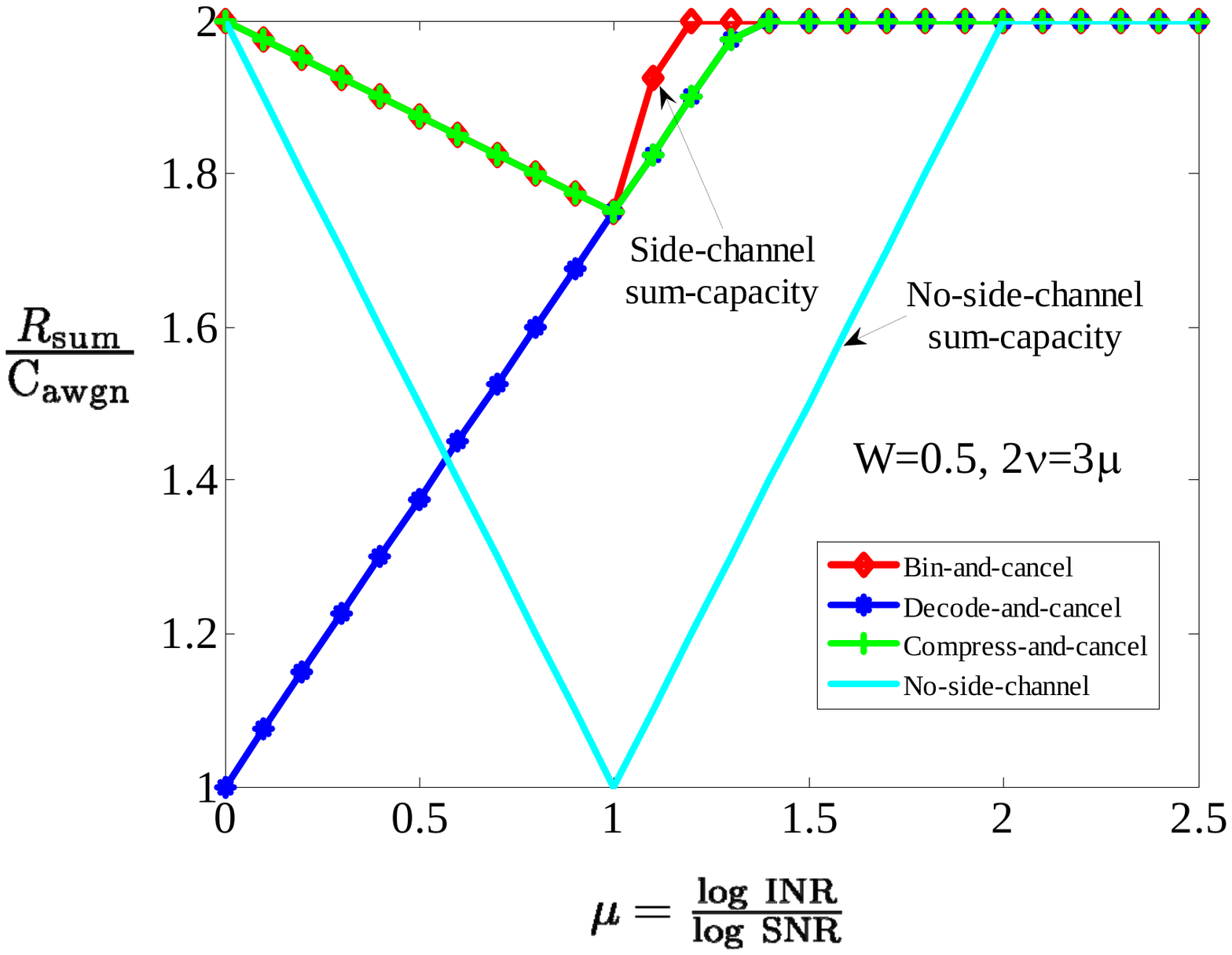}}}
  \subfigure[$W\geq1$ and $\nu<\mu$]
{\label{asmg2}{\includegraphics[width=0.55\textwidth,trim = 30mm
      30mm 20mm 40mm, clip]{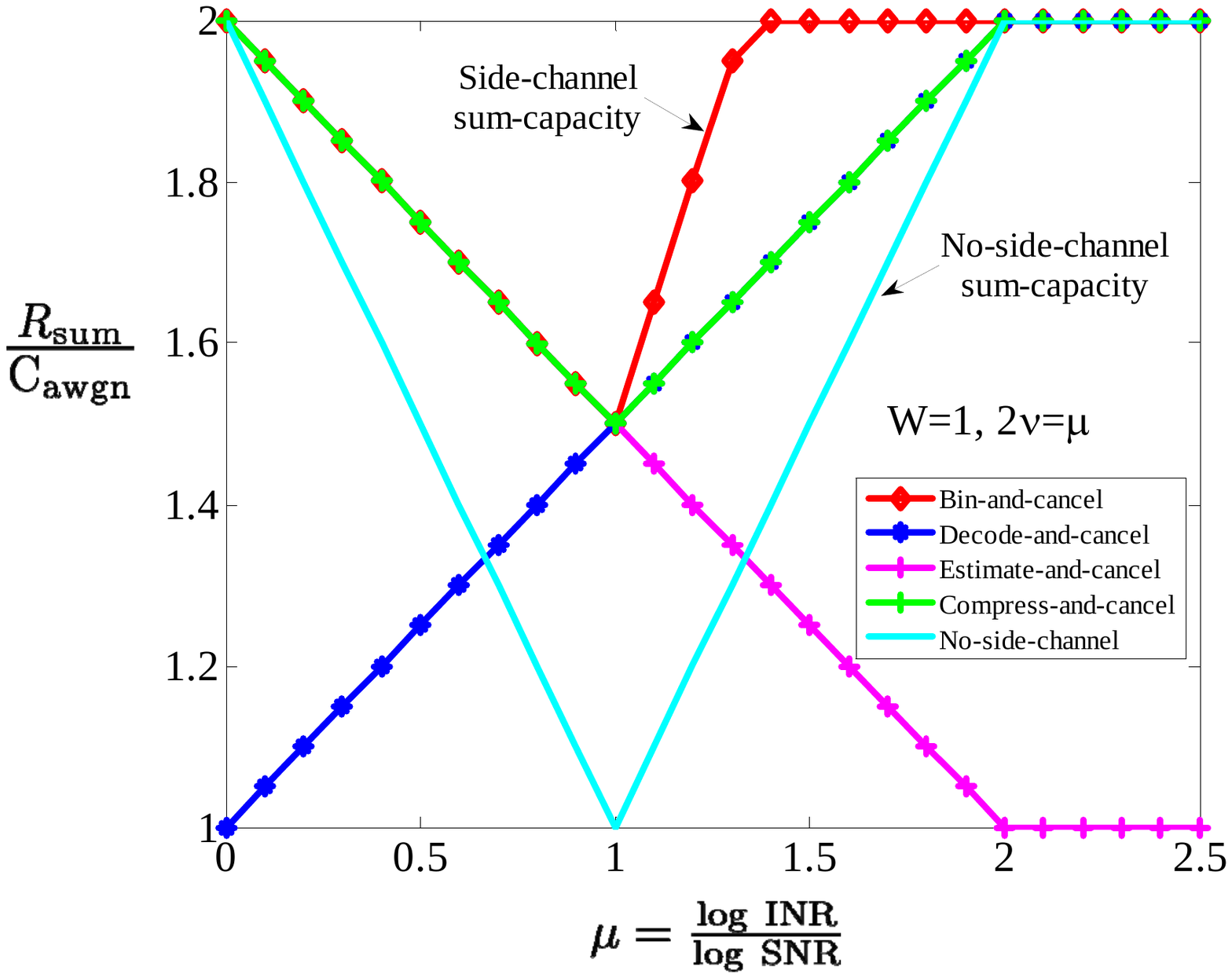}}}
\end{center}
\caption{Asymptotic multiplexing gain for optimal scheme and suboptimal schemes versus no-side-channel sum-capacity. The suboptimal schemes are compress-and-cancel, decode-and-cancel and estimate-and-cancel.}
\label{AMG}
\end{figure}

The asymptotic multiplexing gain of each scheme is plotted in Figure~\ref{AMG}. When $\nu<\mu$, the multiplexing gain of EC is a decreasing function of $\mu$. The intuition is that the side-channel can not provide a better estimate of the interfering signal when the side-channel is worse than the interference link~(i.e., main-channel between M1 and M2). While for DC, when $\frac{\nu}{\mu}$ is fixed, the multiplexing gain of DC will scale with interference level. However, DC has poor performance when the interference level is very low. The reasons behind this is that the optimal power allocated to the side-channel is inversely proportional to the interference level as derived in \cite{Jingwen}. Therefore, uplink Node M1 needs to allocate most of its power to the side-channel to decode the interfering signal in the very weak interference regime, with little power left to send its own data to the base station through the main-channel. However, as the interference level increases, the optimal power allocated to the side-channel drops rapidly, thus improving the performance of DC.
\subsection{Finite SNR Multiplexing Gain and Optimal Power Allocation}
In this section, we show numerically the performance of the proposed schemes for finite SNR values. Figure~\ref{finiteSNR} depicts the multiplexing gain of each scheme when $\mathrm{SNR}=15$~dB, $W=1$ and $\nu=\mu$, and Figure~\ref{OPA} shows the corresponding optimal power allocated to the side-channel which maximizes achievable sum-rate of each proposed scheme. In Figure~\ref{finiteSNR}, BC, DC, CC and EC are shown to achieve peak improvement over the no-side-channel case of 1.57 and 1.51, 1.41 and 1.22, respectively. In addition, we can observe that while EC usually surpasses DC in the weak interference regimes, DC dominates in the strong interference regimes. The intuition behind this result is that when two mobile nodes are very close to each other which leads to strong INI, M2 can better decode the information from M1 by the side-channel. However, when the two mobile nodes are not so close resulting in weak INI, a statistical estimate of the interfering signal by the side-channel is better than decoding the signal to help cancel INI. For the percentage of the optimal power allocated
to the side-channel of each scheme, only a small portion of transmit power needs to be allocated to
the side-channel in the regime where the proposed schemes offer better performance.
\begin{figure}[t!] 
  \centering
   \includegraphics[width=0.5\textwidth,trim = 40mm
      35mm 30mm 35mm, clip]{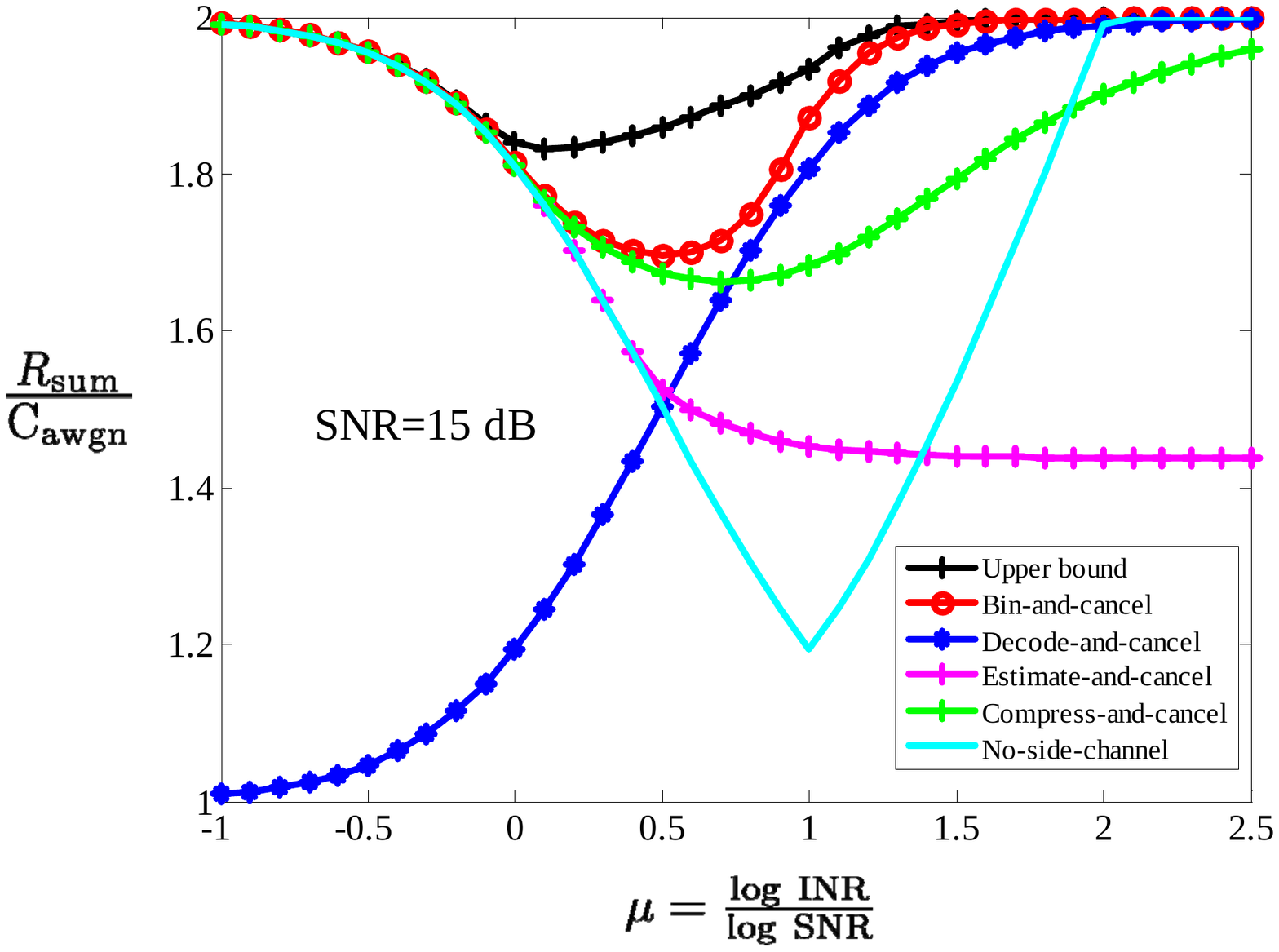}
  \caption{Multiplexing gain of proposed schemes versus no-side-channel achievable schemes for finite SNR when $W=1$ and $\nu=\mu$.}
\label{finiteSNR}
\end{figure}
\begin{figure}[t!] 
  \centering
   \includegraphics[width=0.5\textwidth,trim = 40mm
      35mm 30mm 38mm, clip]{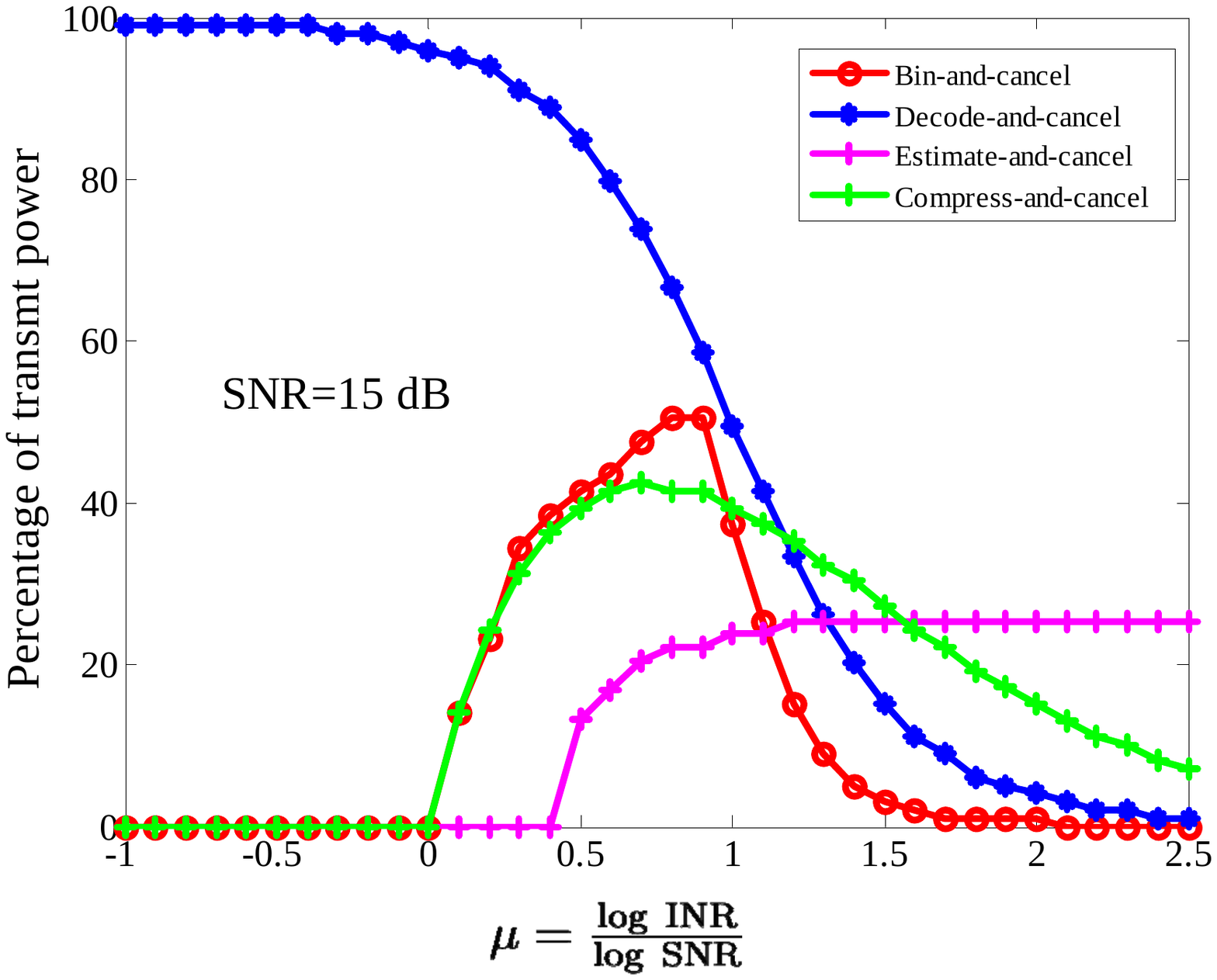}
  \caption{Optimal power allocated to the side-channel which maximizes achievable sum-rates of proposed schemes for finite SNR when $W=1$ and $\nu=\mu$}
\label{OPA}
\end{figure}
\section{Conclusion}
In a three-node full-duplex network where an infrastructure node communicates with half-duplex mobile nodes for both uplink and downlink simultaneously in the same band, one of the key challenges is the inter-node interference. We identify the availability of multi-radio interfaces on current mobile devices to create a wireless side-channel
that allows for new approaches to mitigate inter-node interference. Therefore we propose distributed full-duplex architecture via wireless side-channels for advanced interference management.

In this paper, we presented four distributed full-duplex inter-node interference cancellation schemes by leveraging a device-to-device side-channel for improved
interference cancellation. We characterize the bounds on the capacity region of side-channel assisted three-node network and show that bin-and-cancel scheme can achieve within one bit of the capacity region for all values of channel parameters and is asymptotically optimal. The other three schemes are also analyzed and shown to be optimal in specific regimes. Both analytical and numerical results demonstrate the factors that dominate the performance of
the proposed schemes, which contributes to guiding the design of transceiver architecture.
\section{Appendix}
\subsection{Capacity Region of a Multiple Access Channel with Side-Channel} \label{Append1}
\begin{lemma}
The capacity region of a multiple-access channel with side-channel in Figure~\ref{fig.MACsc} is
\begin{small}
\begin{gather}
\begin{aligned}
R_{1}&\leq I(X_1;Y_1|X_{2})\nonumber \\
R_{2}&\leq I(X_{2};Y_1|X_{1})+I(X_3;Y_3)\nonumber \\
R_{1}+R_{2}&\leq I(X_{1},X_{2};Y_1)+I(X_3;Y_3),
\end{aligned}
\end{gather}
\end{small}
for some $p(x_1)p(x_{2},x_3).$
\end{lemma}
 \begin{figure}[t!] 
  \centering
    \includegraphics[width=0.4\textwidth,trim = 70mm
      90mm 80mm 70mm, clip]{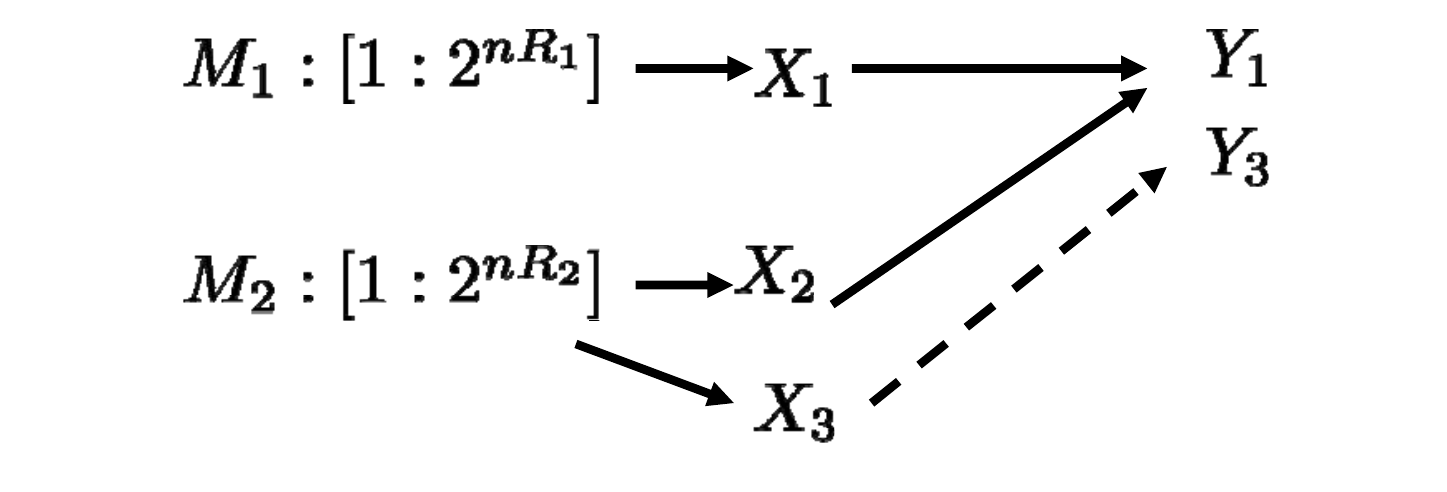}
  \caption{Multiple access channel with side-channel.}
\label{fig.MACsc}
\end{figure}
\begin{proof}
We first give the outline of achievability.
\begin{enumerate}
\item Codebook generation: Fix $p(x_1)$ and $p(x_2,x_3)$ that achieves the lower bound. Randomly and independently generate $2^{nR_1}$ sequence $x_1^n(m_1)$, $m_1\in[1:2^{nR_{1}}]$, each i.i.d. according to distribution $\prod_{i=1}^{n}p_{X_1}(x_{1i}).$ Randomly and independently generate $2^{nR_3}$ sequence $x_3^n(l), l\in[1:2^{nR_3}],$ each i.i.d. according to distribution $\prod_{i=1}^{n}p_{X_3}(x_{3i}).$
    For each $l$, randomly and conditionally independently generate $2^{nR_{2}}$ sequence $x_2^n(m_2|l)$, $m_2\in[1:2^{nR_2}]$, each i.i.d. according to distribution $\prod_{i=1}^{n}p_{X_2|X_3}(x_{2i}|x_{3i})$. Partition the set $[1:2^{nR_{2}}]$ into $2^{nR_3}$
equal size bins, $\mathcal{B}(l)=[(l-1)2^{n(R_{2}-R_3)}+1:l2^{n(R_{2}-R_3)}], l\in[1:2^{nR_3}]$. The codebook and bin assignments are revealed to all parties.
\item Encoding: Upon observing $m_2\in[1:2^{nR_{2}}],$ assign $m_2$ to bin $\mathcal{B}(l)$. The encoders send $x_1^n(m_1), x_2^n(m_2|l)$ over the main-channel, and $x_3^n(l)$ over the side-channel in $n$ blocks.
\item Decoding: Upon receiving $Y_3$, the receiver for the side-channel declares that $\hat{l}$ is sent if it is the unique message such that $(x_3^n(\hat{l}),Y_3^n)\in\mathcal{T}_{\epsilon}^n$; otherwise, it declares an error. Then upon receiving $Y_1$, the receiver for the main-channel declares that $(\hat{m_1}, \hat{m_2})$ are sent if it is the unique pair such that $(x_1^n(\hat{m_1}),x_2^n(\hat{m_2}|\hat{l}),Y_1^n)\in\mathcal{T}_{\epsilon}^n$ and $\hat{m_2}\in\mathcal{B}(\hat{l})$; otherwise, it declares an error.
\end{enumerate}
Therefore, the following rate region is achievable
\begin{small}
\begin{gather}
\begin{aligned}
R_{1}&\leq I(X_1;Y_1|X_{2})\nonumber \\
R_{2}&\leq I(X_{2};Y_1|X_{1})+I(X_3;Y_3)\nonumber \\
R_{1}+R_{2}&\leq I(X_{1},X_{2};Y_1)+I(X_3;Y_3),
\end{aligned}
\end{gather}
\end{small}
for some $p(x_1)p(x_{2},x_3).$

The converse comes from cut-set upper bound.
\begin{small}
\begin{eqnarray}
&&\!\!\!\!\!\!n(R_2)\leq I(X_2^n,X_3^n;Y_1^n,Y_3^n|X_1^n)\\
&&\!\!\!\!\!\!=I(X_2^n;Y_1^n,Y_3^n|X_1^n)+I(X_3^n;Y_1^n,Y_3^n|X_1^n,X_2^n) \\
&&\!\!\!\!\!\!=I(X_2^n;Y_1^n|X_1^n)+I(X_2^n;Y_3^n|X_1^n,Y_1^n)\\
&&\!\!\!\!\!\!+H(X_3^n|X_1^n,X_2^n)-H(X_3^n|X_1^n,X_2^n,Y_1^n,Y_3^n)\\
&&\!\!\!\!\!\!=I(X_2^n;Y_1^n|X_1^n)+I(X_2^n;Y_3^n|X_1^n,Y_1^n) \label{a}\\
&&\!\!\!\!\!\!\leq I(X_2^n;Y_1^n|X_1^n)+I(X_3^n;Y_3^n), \label{b}
\end{eqnarray}
\end{small}
where (\ref{a}) follows that $X_3^n$ is a function of $X_2^n$, thus \begin{small}\[H(X_3^n|X_1^n,X_2^n)=H(X_3^n|X_1^n,X_2^n,Y_1^n,Y_3^n),\]\end{small} and (\ref{b}) follows that
\begin{small}
\begin{gather}
\begin{aligned}
I(X_2^n;Y_3^n|X_1^n,Y_1^n)&=H(Y_3^n|X_1^n,Y_1^n)-H(Y_3^n|X_1^n,X_2^n,Y_1^n)\\
&\leq H(Y_3^n)-H(Y_3^n|X_3^n)\\
&\leq I(X_3^n|Y_3^n).
\end{aligned}
\end{gather}
\end{small}
Similarly, we can also prove that
\begin{small}
\begin{gather}
\begin{aligned}
R_{1}&\leq I(X_1;Y_1|X_{2})\nonumber \\
R_{1}+R_{2}&\leq I(X_{1},X_{2};Y_1)+I(X_3;Y_3).
\end{aligned}
\end{gather}
\end{small}
For Gaussian MAC with side-channel, the above constraints on $R_1, R_2$ and $R_1+R_2$ will be simultaneously maximized by Gaussian inputs.
\end{proof}
\subsection{Calculation of Achievable Rates of Compress-, estimate-and-cancel schemes} \label{Append2}
For compress-and-cancel scheme, with Wyner-Ziv strategy~\cite{wyner}, the achievable rate region is given by
\begin{gather}
\begin{aligned}
R_1\leq & I(X_1;Y_1,\widehat{X_2})\\
R_2\leq&I(X_2;Y_2)\\
\text{s.t.}~~&I(X_2;\widehat{X_2}|Y_1)\leq I(X_3;Y_3),
\end{aligned}
\end{gather}
for distribution product $p(x_1)p(x_2)p(x_3)p(\widehat{x_2}|x_2,y_1)$.

In the Gaussian case, all inputs are assumed to be i.i.d. Gaussian distributed, satisfying $X_1\sim\mathcal{N}(0,P_1), X_2\sim\mathcal{N}(0,\bar{\lambda}P_2)$
and $X_{3}\sim\mathcal{N}(0,\lambda P_2)$, respectively, where $\lambda\in[0,1],\bar{\lambda}+\lambda=1$.
Assuming $\widehat{X_2}$ is a Gaussian quantized version of $X_2$:
\begin{gather}
\begin{aligned}
\widehat{X_2}=X_2+Z_q,
\end{aligned}
\end{gather}
where $Z_q$ is the quantization noise with Gaussian distribution $\mathcal{N}(0,\sigma_q^2)$.
We have
\begin{gather}
\begin{aligned}
R_{\widehat{X_2}|Y_1}&=I(X_2;\widehat{X_2}|Y_1)\\
&=h(\widehat{X_2}|Y_1)-h(\widehat{X_2}|Y_1,X_2)\\
&=W_m\mathrm{log}\left(1+\frac{\bar{\lambda}P_2(\gamma_1P_1+\sigma^2)}{(\gamma_1P_1+\gamma_{21}\bar{\lambda}P_2+\sigma^2)\sigma_q^2}\right),
\end{aligned}
\end{gather}
and
\begin{gather}
\begin{aligned}
I(X_3;Y_3)\leq W_s\mathrm{log}\left(1+\frac{\gamma_3\lambda P_2}{W\sigma^2}\right).
\end{aligned}
\end{gather}
Thus subjecting to the constraint, we have
\begin{gather}
\sigma_q^2\geq\frac{\bar{\lambda}P_2(\gamma_1P_1+\sigma^2)}{(\gamma_1P_1+\gamma_{21}\bar{\lambda}P_2+\sigma^2)\left[\left(1+\frac{\gamma_3\lambda P_2}{W\sigma^2}\right)^W-1\right]}\label{CCR11}.
\end{gather}
And
\begin{gather}
\begin{aligned}
R_1&\leq I(X_1;Y_1,\hat{X_2})\\
&=W_m\mathrm{log}\left(1+\frac{\gamma_1P_1(\bar{\lambda}P_2+\sigma_q^2)}{\gamma_{21}\bar{\lambda}P_2\sigma_q^2+\sigma^2(\bar{\lambda}P_2+\sigma_q^2)}\right). \label{CCR1}
\end{aligned}
\end{gather}
Since $R_1$ is a decreasing function of $\sigma_q^2$, now we can calculate $R_1$ by substituting (\ref{CCR11}) into (\ref{CCR1}).
And $R_2$ is given as
\begin{gather}
\begin{aligned}
R_2&\leq I(X_2;Y_2)\\
&=C\left(\bar{\lambda}\rm SNR_2\right).
\end{aligned}
\end{gather}

For estimate-and-cancel, since $\mathbb{E}(|X_2|^2)\leq \frac{P_2}{(1+K)^2}$, we choose $X_1$ and $X_2$ from independent Gaussian codebooks with $X_1\sim\mathcal{N}(0,P_1), X_2\sim\mathcal{N}\left(0,\frac{P_2}{(1+K)^2}\right)$.
The correlation coefficient between $X_2$ and $X_3$ is one, and $X_{3}\sim\mathcal{N}\left(0,\frac{K^2P_2}{(1+K)^2}\right)$.
For Gaussian channels, by substituting the channel parameters and power constraint, the achievable rate region can be calculated
\begin{gather}
\begin{aligned}
R_1 &\leq  I(X_1;Y_1,Y_3)\\
& = h(Y_1, Y_3)-h(Y_1,Y_3|X_1)\\
& = C\left(\frac{\mathrm {SNR_1}\left(1+\frac{K^2 \rm SNR_{side}}{(1+K)^2}\right)}{1+\frac{K^2 \mathrm{SNR_{side}}}{(1+K)^2}+\frac{\rm INR}{(1+K)^2}}\right)\\
R_2 & \leq  I(X_2;Y_2)\\
& = C\left(\frac{\rm SNR_2}{(1+K)^2}\right).
\end{aligned}
    \end{gather}


\begin{thebibliography}{10}
\providecommand{\url}[1]{#1}
\csname url@samestyle\endcsname
\providecommand{\newblock}{\relax}
\providecommand{\bibinfo}[2]{#2}
\providecommand{\BIBentrySTDinterwordspacing}{\spaceskip=0pt\relax}
\providecommand{\BIBentryALTinterwordstretchfactor}{4}
\providecommand{\BIBentryALTinterwordspacing}{\spaceskip=\fontdimen2\font plus
\BIBentryALTinterwordstretchfactor\fontdimen3\font minus
  \fontdimen4\font\relax}
\providecommand{\BIBforeignlanguage}[2]{{%
\expandafter\ifx\csname l@#1\endcsname\relax
\typeout{** WARNING: IEEEtran.bst: No hyphenation pattern has been}%
\typeout{** loaded for the language `#1'. Using the pattern for}%
\typeout{** the default language instead.}%
\else
\language=\csname l@#1\endcsname
\fi
#2}}
\providecommand{\BIBdecl}{\relax}
\BIBdecl

\bibitem{evan}
E.~Everett, M.~Duarte, C.~Dick, and A.~Sabharwal, ``Empowering full-duplex
  wireless communication by exploiting directional diversity,'' in
  \emph{Proc. IEEE Asilomar Conf. Signals, Syst., Comput.}, Nov. 2011.

\bibitem{sahai2011pushing}
A.~Sahai, G.~Patel, and A.~Sabharwal, ``Pushing the limits of full-duplex:
  Design and real-time implementation, \url{http://arxiv.org/abs/1107.0607},'' in
  \emph{Rice University Technical Report TREE1104}, June 2011.

\bibitem{duarte2012design}
M.~Duarte, A.~Sabharwal, V.~Aggarwal, R.~Jana, K.~Ramakrishnan, C.~Rice, and
  N.~Shankaranarayanan, ``Design and characterization of a full-duplex
  multi-antenna system for wifi networks,'' \emph{Submitted to Veh. Technol., IEEE Trans.}, vol. \url{http://arxiv.org/abs/1210.1639}, Oct.
  2012.

\bibitem{Choi}
J.~I. Choi, M.~Jain, K.~Srinivasan, P.~Levis, and S.~Katti, ``Achieving single
  channel, full duplex wireless communication,'' in \emph{Proc. Sixteenth Annual International Conf. Mobile Comput.
  Netw.}, ser. MobiCom '10. New
  York, NY, USA: ACM.

\bibitem{duarte2011experiment}
M.~Duarte, C.~Dick, and A.~Sabharwal, ``Experiment-driven characterization of
  full-duplex wireless systems,'' \emph{To appear: Wireless Commun., IEEE
  Trans.}.

\bibitem{radunovic2010rethinking}
B.~Radunovic, D.~Gunawardena, P.~Key, A.~Proutiere, N.~Singh, V.~Balan, and
  G.~Dejean, ``Rethinking indoor wireless mesh design: Low power, low
  frequency, full-duplex,'' in \emph{Wireless Mesh Netw. (WIMESH 2010), 2010
  Fifth IEEE Workshop}. IEEE, 2010.

\bibitem{Everett12Thesis}
E.~Everett, A.~Sahai, and A.~Sabharwal, ``Passive self-interference suppression
  for full-duplex infrastructure nodes,'' \emph{Submitted to Wireless
  Commun., IEEE Trans.}, vol. \url{http://arxiv.org/abs/1302.2185},
  2013.
\bibitem{achal}
A.~Sahai, V.~Aggarwal, M.~Yuksel and Sabharwal, ``Capacity of All Nine Models of Channel Output Feedback for the Two-User Interference Channel,'' 
\emph{Inf. Theory, IEEE Trans.}, vol.~59, Nov. 2013. 
\bibitem{Smallcell}
Small cell forum: [Online]. Available: \url{http://www.smallcellforum.org/}.

\bibitem{Jingwen}
J.~Bai and A.~Sabharwal, ``Decode-and-cancel for interference cancellation in
  full-duplex networks,'' in \emph{Proc. IEEE Asilomar Conf. Signals, Syst., Comput.}, Nov 2012.

\bibitem{Zchannel}
L.~Zhou and W.~Yu, ``Gaussian z-interference channel with a relay link:
  Achievability region and asymptotic sum capacity,'' \emph{Inf. Theory,
  IEEE Trans.}, vol.~58, no.~4, pp. 2413--2426, 2012.

\bibitem{han1981new}
T.~Han and K.~Kobayashi, ``A new achievable rate region for the interference
  channel,'' \emph{Inf. Theory, IEEE Trans.}, Jan. 1981.

\bibitem{Car}
A.~Carleial, ``A case where interference does not reduce capacity (corresp.),''
  \emph{Inf. Theory, IEEE Trans.}, vol.~21, Sept. 1975.

\bibitem{D}
R.~Etkin, D.~Tse, and H.~Wang, ``Gaussian interference channel capacity to
  within one bit,'' \emph{Inf. Theory, IEEE Trans.}, vol.~54,
  Dec. 2008.

\bibitem{wyner}
A.~Wyner and J.~Ziv, ``The rate-distortion function for source coding with side
  information at the decoder,'' \emph{Inf. Theory, IEEE Trans.}, vol.~22, 1976.

\bibitem{Cog1}
S.~Seyedmehdi, J.~Jiang, Y.~Xin, and X.~Wang, ``An improved achievable rate
  region for causal cognitive radio,'' in \emph{Proc. IEEE ISIT 2009},
  July 2009.

\bibitem{Cog2}
Y.~Cao and B.~Chen, ``Interference channel with one cognitive transmitter,'' in
  \emph{Proc. IEEE Asilomar Conf. Signals, Syst., Comput.}, Oct. 2008.

\end{thebibliography}
\bibliographystyle{IEEEtran}
 \begin{IEEEbiographynophoto}{Jingwen Bai}
received her B.E degree in Electronic Science and Technology from Beijing University of Posts and Telecommunications, Beijing, China, in 2011, and the M.S degree in Electrical and Computer Engineering from Rice University, Houston, TX, in 2013. She is currently working toward her Ph.D. at Rice University, Houston, TX under the guidance of Dr. Ashutosh Sabharwal. Her research interest includes information theory and wireless communication.
 \end{IEEEbiographynophoto}
 \begin{IEEEbiographynophoto}{Ashutosh Sabharwal}
(S$^\prime$91-M$^\prime$99-SM$^\prime$04-F$^\prime$14) received the B.Tech. degree from
the Indian Institute of Technology, New Delhi, in 1993 and the M.S. and Ph.D.
degrees from The Ohio State University, Columbus, in 1995 and 1999, respectively.
He is currently a Professor in the Department of Electrical and
Computer Engineering, Rice University, Houston, TX. His research interests
include the areas of information theory and communication algorithms for wireless systems.
Dr. Sabharwal was the recipient of Presidential Dissertation Fellowship
Award in 1998, and the founder of WARP project (http://warp.rice.edu).
 \end{IEEEbiographynophoto}
\end{document}